\documentclass[11pt]{article}
\usepackage{amsmath,amsthm,amssymb,amsfonts}
\newcommand{\field}[1]{\mathbb{#1}}
\newcommand{\remove}[1]{}
\usepackage{longtable}
\usepackage{complexity}
\usepackage{hyperref}

\newtheorem{thm}{Theorem}[section]
\newtheorem{claim}[thm]{Claim}
\newtheorem{lem}[thm]{Lemma}
\newtheorem{define}[thm]{Definition}
\newtheorem{cor}[thm]{Corollary}
\newtheorem{obs}[thm]{Observation}

\def\F{{\mathbb{F}}}

\def\Z{{\mathbb{Z}}}
\def\N{{\mathbb{N}}}

\def\1{\mathbf 1}

\def\_{\,\,\,\,\,}

\newcommand{\sop}{\Sigma\Pi\left(\Sigma\Pi\right)^{[s]}}
\newcommand{\sopone}{\Sigma\Pi\left(\Sigma\Pi\right)^{[1]}}
\newcommand{\soptwo}{\Sigma\Pi\left(\Sigma\Pi\right)^{[2]}}
\newcommand{\sopm}{\Sigma\Pi\left(\Sigma\Pi\right)^{[N^{{\mu}}]}}
\newcommand{\nwm}{NW_{n,{\mu}}}
\newcommand{\nwa}{NW_{{a'},  { \mu'}}}
\newcommand{\nwb}{NW'_{{a'},  { \mu'}}}

\def\dim{\mathsf{Dim}}
\def\h{\mathsf{Hom}}
\def\es{\mathsf{ESYM}}
\begin{document}
\def\e{{\mathbb E}}
\def\lm{\mathsf{Lead\mbox{-}Mon}}

\mathchardef\mhyphen="2D 

\title{Sums of products of polynomials in  few variables : lower bounds and polynomial identity testing}
\author{Mrinal Kumar  \and Shubhangi Saraf}
\author{Mrinal Kumar\thanks{Department of Computer Science, Rutgers University.
Email: \texttt{mrinal.kumar@rutgers.edu}. Research supported in part by NSF grant CCF-1253886.}\and
Shubhangi Saraf\thanks{Department of Computer Science and Department of Mathematics, Rutgers University.
Email: \texttt{shubhangi.saraf@gmail.com}. Research supported by NSF grant CCF-1350572.}}

\date{}
\maketitle
\abstract{
We study the complexity of representing polynomials as a sum of products of polynomials in few variables. More precisely, we study representations of the form $$P = \sum_{i = 1}^T \prod_{j = 1}^d Q_{ij}$$
such that each $Q_{ij}$ is an arbitrary polynomial that depends on at most $s$ variables. 

We prove the following results. 

\begin{itemize}
\item Over fields of characteristic zero, for every constant $\mu$ such that $0 \leq \mu < 1$, we give an explicit family of polynomials $\{P_{N}\}$, where $P_{N}$ is of degree $n$ in $N = n^{O(1)}$ variables, such that any representation of the above type for $P_{N}$ with $s = N^{\mu}$  requires $Td \geq n^{\Omega(\sqrt{n})}$. This strengthens a recent result of Kayal and Saha~\cite{KayalSaha14} which showed similar lower bounds for  the model of sums of products of linear forms in few variables. It is known that any asymptotic improvement in the exponent of the lower bounds (even for $s = \sqrt{n}$) would separate $\VP$ and $\VNP$~\cite{KayalSaha14}. 
\item We obtain a deterministic subexponential  time blackbox polynomial identity testing (PIT) algorithm for circuits computed by the above model when $T$ and the individual degree of each variable in $P$ are at most $\log^{O(1)} N$ and $s \leq N^{\mu}$ for any constant $\mu < 1/2$. We get quasipolynomial running time when $s < \log^{O(1)} N$.  
The PIT algorithm is obtained by combining our lower bounds with the hardness-randomness tradeoffs developed in~\cite{DSY09, KI04}.  To the best of our knowledge, this is the first nontrivial PIT algorithm for this model (even for the case $s=2$), and the first nontrivial PIT algorithm obtained from lower bounds for small depth circuits.~\footnote{In a  recent independent work, Forbes~\cite{Forbes-personal} does blackbox identity testing for another subclass of depth four circuits using shifted partial derivative based methods. To the best of our understanding, the results in these two papers are incomparable even though both rely on similar techniques.}
\end{itemize}
}

\newpage

\section{Introduction}

Arithmetic circuits are the most natural model of computation for a wide variety of algebraic problems such as matrix multiplication, computing fast fourier transforms etc. The problem of proving lower bounds for arithmetic circuits is one of the most fundamental and interesting problems in complexity theory. Proving superpolynomial lower bounds for general arithmetic circuits would resolve the $\VP$ versus $\VNP$ conjecture~\cite{Valiant79}, the algebraic analog of the $\P$ vs $\NP$ conjecture. This is one of the holy grails of complexity theory and has received a lot of attention, since it is a more structured and potentially easier question to understand and analyse than the $\P$ vs $\NP$ problem . 

The intimately related problem of polynomial identity testing (PIT) is the problem of testing if a polynomial, given as an arithmetic circuit is identically zero. In the setting where the algorithm cannot look inside the circuit, but only has access to evaluations of the circuit, the problem is referred to as blackbox PIT. There is a very simple randomized algorithm for this problem - simply evaluate the polynomial at a random point from a large enough domain. With very high probability, a nonzero polynomial will have a nonzero evaluation~\cite{Schwartz80, Zippel79}. It is a very important and fundamental question to derandomize the above algorithm. In a seminal work, Kabanets and Impagliazzo~\cite{KI04} showed that the problem of proving  lower bounds for arithmetic circuits and the problem of derandomizing identity testing are essentially equivalent!  

These two problems have occupied a central position in complexity theory and despite much attention, our understanding of general arithmetic circuits is still very limited.  Thus there has been a great deal of effort in understanding the complexity of restricted classes of arithmetic circuits in an attempt to obtain a better understanding of the general problem.  Low depth arithmetic circuits in particular are one such well studied class. 

\paragraph{Lower bounds for homogeneous low depth arithmetic circuits.}
The last few years have seen a tremendous amount of exciting progress on the problems of ``depth reduction" of general arithmetic circuits to low depth arithmetic circuits, and of proving lower bounds for low depth arithmetic circuits. Using depth reduction techniques~\cite{VSBR83, AV08, koiran, Tavenas13} it was shown that $N^{\omega(\sqrt n)}$ lower bounds (for polynomials in $N$ variables and of degree $n$) for just homogeneous depth 4 arithmetic circuits of bottom fan-in $\sqrt n$ would suffice to separate $\VP$ from $\VNP$ and imply superpolynomial lower bounds for general arithmetic circuits. At the same time there was a very exciting line of works proving $N^{\Omega(\sqrt n)}$ lower bounds for the same model of arithmetic circuits (and in fact for even the more general class of homogeneous depth 4 arithmetic circuits with no restriction on bottom fan-in)~\cite{GKKS12, FLMS13, KSS13, KS-formula, KLSS14, KS-full}.

\paragraph{Lower bounds for non-homogeneous low depth arithmetic circuits.} Despite all this remarkable progress, and some very strong lower bounds for homogeneous low depth arithmetic circuits, in the nonhomogenous world much less is understood. Only mild lower bounds are known when we drop the condition of homogeneity,  even for very simple classes of low depth arithmetic circuits. For depth 3 circuits over fields of characteristic 0, only quadratic lower bounds known~\cite{SW01, Shp01}, and there has been no progress on this question in more than a decade now. 

In a beautiful depth reduction result over fields of characteristic 0, Gupta et al~\cite{GKKS13} showed that $N^{\omega(\sqrt n)}$ lower bounds (for polynomials in $N$ variables and of degree $n$) for the class of non-homogeneous {\it depth 3} circuits would already separate $\VP$ from $\VNP$. It was recently observed by Kayal and Saha~\cite{KayalSaha14}~\footnote{They attribute the observation to Ramprasad Saptharishi.} that in fact it suffices to prove such lower bounds for depth 3 circuits with bottom fan-in $\sqrt n$. 

Till recently (in particular till the work of~\cite{KayalSaha14}), the best known lower bounds for depth 3 circuits even with bottom fan-in 2 were still just quadratic. In a very nice recent result, Kayal and Saha~\cite{KayalSaha14} showed an exponential lower bound for depth 3 circuits over fields of characteristic 0, whose bottom fan-in is at most $N^{\mu}$, where $N$ is the number of variables and $0 \leq \mu < 1$ is an arbitrary constant. More precisely, they prove the following. 

\begin{thm}[Kayal-Saha~\cite{KayalSaha14}]~\label{thm: KSaha}
Let $\F$ be a field of characteristic zero. Then, for every constant $0 \leq \mu < 1$ there is a family $\{P_N\}$ of degree $n$ polynomials in $N = n^{O_{\mu}(1)}$ variables over $\F$ in $\VNP$ such that any depth three circuit of bottom fan-in at most $N^{\mu}$ computing $P_N$ has top fan-in at least $N^{\Omega_{\mu}{(\sqrt{n})}}$. 
\end{thm}

\paragraph{Our Model:} In this work, we consider the model of sums of products of polynomials in few variables. More formally, we consider representations of polynomials $P$ (degree $n$ in $N = n^{O(1)}$ variables) in the form
\begin{equation}~\label{def:model intro}
P = \sum_{i=1}^T \prod_{j = 1}^d Q_{ij}
\end{equation}
where each $Q_{ij}$ is an arbitrary polynomial (of arbitrarily high degree) in at most $s$ variables. We call this the model of $\sop$ circuits. 

Observe that the model is more general than that considered in~\cite{KayalSaha14}. The model in ~\cite{KayalSaha14} corresponds to sums of products of {\it linear forms} in few variables. In our case, the $Q_{ij}$ no longer have to be linear forms, but can be general polynomials of arbitrarily high degree. Prior to this work, even for the case when $s = 2$, there were no nontrivial lower bounds known for this model. 




$\sop$ circuits for $s \geq 2$ can also be seen as a generalization of the model of sums of products of univariate polynomials (which corresponds to $\sop$ circuits with $s=1$), which has been very well studied in the arithmetic circuit complexity literature. Lower bounds for $\sopone$ circuits follow from works of Nisan~\cite{Nisan91} and Saxena~\cite{S07}. Over the last few years, there have been some very nice results giving quasipolynomial time blackbox identity testers for $\sopone$ circuits~\cite{ForbesS13, FS13, ASS13}.
$\sop$ circuits can also be seen as a generalization of the widely studied model of diagonal circuits, since polynomials computable by diagonal circuits can be represented as a $\sopone$ circuit without much blow up in the size of the representation~\cite{S07}. 

Although $\sopone$ circuits seem fairly well understood from the point of view of lower bounds and derandomization of polynomial identity testing, if one considers the model of sums of products of bivariate polynomials ($\soptwo$ circuits), then our understanding changes completely. Although only seemingly a mild generalization of $\sopone$ circuits, the known proof techniques
for lower bounds for $\sopone$ circuits (which were proved using {\it evaluation dimension} techniques of~\cite{Nisan91, Raz06}) seem to completely break down in this setting. Thus, studying this model seems like an interesting next step towards understanding non-homogeneous small depth algebraic computation. 
As far as we are aware there are also (not surprisingly) no nontrivial PIT results for the model. 
We are now ready to state our results.
\subsection{Our results} 
\paragraph{Lower bounds : } We show an exponential lower bound for the model of $\sop$, when $s$ is at most $N^{\mu}$ for any constant $0 \leq \mu < 1$ ($N$ is the number of variables). More precisely, we show the following. 

\begin{thm}~\label{thm:mainthm intro}
Let $\F$ be a field of characteristic zero and $\mu$ be any constant such that $0 \leq \mu < 1$. There exists a family $\{P_N\}$ of polynomials over $\F$ in $\VNP$, where $P_N$ is of degree $n$  in $N = n^{O_{\mu}(1)}$ variables, such that for any representation of $P_N$ of the form
$$P_N = \sum_{i = 1}^T\prod_{j = 1}^{d} Q_{ij}$$
where each $Q_{ij}$ is polynomial in at most  $s =N^{\mu}$ variables, it must be true that   $$T\cdot d \geq n^{\Omega_{\mu}(\sqrt{n})}$$
\end{thm}

Given the depth reduction results of~\cite{GKKS13} and the observation mentioned earlier from~\cite{KayalSaha14}, it is known that any asymptotic improvement in the exponent of the lower bound (even for $s = O(\sqrt{n})$) would imply $\VNP$ is different from $\VP$.

As discussed in the introduction, even though this model seems a natural generalization of the model of sums of products of univariate polynomials, our lower bound technique is very different from those used in proving lower bounds for sums of products of univariates.  Our lower bound proof is based on ideas developed in the course of investigating homogeneous depth four arithmetic circuits~\cite{KLSS14, KS-full}.

\paragraph{Blackbox PIT : } We also consider the problem of PIT for the model of $\sop$ circuits. For general sums of products of even bivariate polynomials, this question seems quite difficult, and as of now we are not even able to obtain subexponential time PIT. However, as a consequence of our lower bounds and by suitably adapting hardness randomness tradeoffs for arithmetic circuits developed in~\cite{KI04} and~\cite{DSY09}, we are able to obtain PIT results in the setting where the top fan-in of the circuit is bounded, and when we have the promise that the circuit
computes a polynomial of low individual degree.

Our understanding of blackbox PIT for depth four circuits is very limited, and the results known are in very restricted settings. Saraf and Volkovich~\cite{SarafV11} gave blackbox PIT algorithms for multilinear depth 4 circuits with bounded top fan-in. To the best of our knowledge, the idea in~\cite{SarafV11} does not extend to the case of non-multilinear depth 4 circuits, even when the individual degree of each of the variables is at most $2$. Recently, Oliveira et al~\cite{OSV14} gave a subexponential time blackbox PIT for all depth four multilinear circuits\footnote{The running time increases with the size of the circuit, and in particular, it is subexponential time for polynomial sized depth four multilinear circuits.}. In the non-multilinear setting, Agrawal et al.~\cite{ASSS12} gave PIT algorithms for constant depth formulas in which the number of {\it occurences} of each variable is bounded. Without going into the technical details, we remark that the notion of {\it bounded occur} is a generalization of the well studied notion of bounded reads. The most closely related results to those in this paper that we are aware of are the recent papers of Gupta~\cite{Gupta14} and Mukhopadhyay~\cite{Mukhopadhyay15}, which give blackbox PIT results for sums of products of low degree polynomials, where the top sum fan-in is bounded and the circuits satisfy  certain algebraic geometric restrictions.

So, the question of getting PIT results for general depth four circuits (even with bounded top and bottom fan-in) remains wide open.  For instance we still do not know any nontrivial PIT results for a sum of constant many products of degree 2 polynomials. Though we still don't know how to deal with this question, when we replace the polynomials of low degree with polynomials of few variables (but of arbitrarily large degree), then we are able to obtain quasipolynomial PIT results. There is one added caveat however, that the final polynomial computed
needs to be of low individual degree (as seems necessary for PIT results obtained from 
the known hardness-randomness tradeoffs for bounded depth circuits~\cite{DSY09}). We now formally state the theorem.

\begin{thm}~\label{thm:mainthm2 intro}
Let $c$ and $\mu$ be arbitrary constants such that $c> 0$ and $0 \leq \mu < 1/2$, and let $\F$ be a field of characteristic zero. Let ${\cal C}$ be the set of polynomials $P$ in $N$ variables and individual degree at most $k$ over $\F$, with the property that $P$ can be expressed as 
$$P = \sum_{i = 1}^T \prod_{j = 1}^d Q_{ij}$$
such that
\begin{enumerate}
\item $T < \log^c N$
\item $k < \log ^c N$
\item $d < N^c$
\item each $Q_{ij}$ depends on at most $N^{\mu}$ variables 
\end{enumerate}
Then, there exists a constant $\epsilon < 1$ dependent only on $c$ and $\mu$, such that there is a hitting set of size $\exp(N^{\epsilon})$ for ${\cal C}$ which can be constructed in time $\exp(N^{\epsilon})$. 

\end{thm}
Moreover, from our proof, it also follows that if each of polynomial $Q_{ij}$  depends only on $\log^{O(1)} N$ variables, then both the size of the hitting set and the time to construct it, are upper bounded by a quasipolynomial function in $N$. 

\paragraph{Organisation of the paper:} We provide an overview of the proofs in Section~\ref{sec:overview}. We describe some definitions and preliminaries in Section~\ref{sec:prelims}. We present the proof of the lower bound in Section~\ref{sec:lower bound}. We describe the application to blackbox PIT in Section~\ref{sec:pit} and conclude with some open problems in Section~\ref{sec:open ques}. 
\section{Proof overview}~\label{sec:overview}
In this section, we provide an overview of the main ideas in proofs of Theorem~\ref{thm:mainthm intro} and Theorem~\ref{thm:mainthm2 intro}.
\subsection{Overview of proof of Theorem~\ref{thm:mainthm intro}}~\label{sec:overview lower bounds}
We restate Theorem~\ref{thm:mainthm intro} for the sake of clarity. \\
{\bf Theorem~\ref{thm:mainthm intro}}~\label{thm:mainthm intro2}
{\it Let $\F$ be a field of characteristic zero and $\mu$ be any constant such that $0 \leq \mu < 1$. There exists a family $\{P_N\}$ of polynomials over $\F$ in $\VNP$, where $P_N$ is of degree $n$  in $N = n^{O_{\mu}(1)}$ variables, such that for any representation of $P_N$ of the form
$$P_N = \sum_{i = 1}^T\prod_{j = 1}^{d} Q_{ij}$$
where each $Q_{ij}$ is polynomial in only  $N^{\mu}$ variables, it must be true that   $$T\cdot d \geq n^{\Omega_{\mu}(\sqrt{n})}$$

}
The key difference between proving the above lower bound and the lower bounds for homogeneous depth four circuits is that the formal degree of the circuit in the above case could be much larger than the degree of the polynomial, which is $n$. In fact, even the fan-in of the product gates at level 2, that is $d$ could be much larger than $n$. Therefore, a straightforward application of homogeneous depth four circuit lower bounds does not seem to work. Our proof is in two steps and at a high level follows the strategy of the lower bound for non-homogeneous depth three circuits with bounded bottom fan-in by Kayal and Saha~\cite{KayalSaha14} with some key differences. 
\begin{itemize}
\item In the first step, we obtain another representation of $P_N$, as $$P_N = \sum_{i = 1}^{Td2^{O(\sqrt{n})}}\prod_{j = 1}^{n} Q_{ij}'$$
where every monomial in each of the $Q_{ij}'$ has {\it support}\footnote{A monomial is said to have support support $s$ if it depends on at most $s$ distinct variables.} at most $s$, although each $Q_{ij}'$ could now depend on all the variables. The key property that we have gained from this transformation is that the fan-in of the product gates at level two is bounded by $n$ now, which is the degree of $P_N$. However, we have no bound on the degree of the $Q_{ij}'$. Moreover,  we have blown up the top fan-in a bit, but we  will be able to tolerate this loss if $s$ is small.
\item In the second step, the strategy can be seen in two stages. If $\mu$ was very small, say $0.001$, then we could have taken advantage of the fact that in the representation obtained in the first step above, the product fan-in is at most $n$ and the support of every monomial in each of the $Q_{ij}'$ is small, to prove an upper bound on the dimension of the space of projected shifted partial derivatives of the above representation. Comparing this dimension with that of our hard polynomial gives us our lower bound. For larger values of $\mu$, we use random restrictions to ensure that all the monomials of {\it large support} in  $Q_{ij}'$ are set to zero. At the end of such a procedure, we are back to the low support case. This step of the proof is closely along the lines of the proof of homogeneous depth four arithmetic circuit lower bounds in~\cite{KLSS14, KS-full} although in the present case, formal degree of the circuit could be as large as $n^2$, which is much larger than the degree of the polynomial $P_N$. For such large formal degrees, in general we do not even know lower bounds for non-homogeneous depth three circuits.  
\end{itemize}
We would like to point out that the first step of the proof above is similar to the homogenization step in the proof of lower bounds for general depth three circuits with bounded bottom fan-in by Kayal and Saha~\cite{KayalSaha14}. The key difference is that while the circuit they obtain at the end of this step is a strictly homogeneous circuit of formal degree $n$, we are unable to get a similar structure. The complication stems from the fact that when $Q_{ij}$ are not affine forms, they could contain monomials of varying degrees. In this case, it seems difficult to obtain a strict homogenization with a small blow up in size. We get around this deficiency by a more subtle analysis in  the second step, where we show a lower bound for a circuit which has a formal degree much larger than the degree of the polynomial being computed, but has some added structure. This step critically uses that the fact that the product fan-in at level two of these circuits is at most $n$, and the support of every monomial in each of the $Q_{ij}'$ is small.

\subsection{Overview of proof of Theorem~\ref{thm:mainthm2 intro}} 
We first restate  Theorem~\ref{thm:mainthm2 intro}. \\
{\bf Theorem~\ref{thm:mainthm2 intro}}~\label{thm:mainthm2 intro 2}{\it
Let $c$ and $\mu$ be arbitrary constants such that $c> 0$ and  $0 \leq \mu < 1/2$, and let $\F$ be a field of characteristic zero. Let ${\cal C}$ be the set of polynomials $P$ in $N$ variables and individual degree at most $k$ over $\F$, with the property that $P$ can be expressed as 
$$P = \sum_{i = 1}^T \prod_{j = 1}^d Q_{ij}$$
such that
\begin{enumerate}
\item $T < \log^c N$
\item $k < \log ^c N$
\item $d < N^c$
\item each $Q_{ij}$ depends on at most $N^{\mu}$ variables 
\end{enumerate}
Then, there exists a constant $\epsilon < 1$ dependent only on $c$ and $\mu$, such that there is a hitting set of size $\exp(N^{\epsilon})$ for ${\cal C}$ which can be constructed in time $\exp(N^{\epsilon})$. 
 }

\vspace{2mm}
The construction of the hitting set is based on the well known idea of using hard functions for derandomization. Our goal is to reduce the number of variables from $N$ to at most $N^{\delta}$ for some constant $\delta < 1$, while maintaining the zeroness/nonzeroness of the polynomial being tested~\cite{KI04, DSY09}. Once we have done this, we take a brute force hitting set of size $\text{(Degree + 1)}^{\text{Number of variables}}$ as given by Lemma~\ref{lem: comb nulls}. 
To reduce the number of variables, we use the framework introduced by Kabanets and Impagliazzo~\cite{KI04}. 

The key technical step of the proof is to show that for a non-zero polynomial $P$ as defined above, if there exists a polynomial $f \in \F[X_1, X_2, \ldots, X_{i-1}, X_{i+1}, X_{i+2}, \ldots, X_{N}]$ such that $X_i-f$ divides $P$, then $f$ can also be expressed as a sum of products of polynomials in few variables of  reasonably small size. This step crucially uses a statement about complexity of  roots of polynomials computed by low depth circuits from~\cite{DSY09}. Therefore, if $f$ is a polynomial which does not have a small representation as a sum of products of polynomials in {\it few} variables, then $X_i - f$ does not divide $P$. This observation guarantees that the construction of hitting sets from hard polynomials given by~\cite{KI04} works for this class of circuits.  


\section{Notation and Preliminaries}~\label{sec:prelims}
We now introduce some notation and preliminary notions that we use in the rest of the paper. 
\paragraph{Computational model : } In this work, we consider the model of sums of products of polynomials in few variables. More formally, we consider representations of polynomials $P$ (degree $n$ in $N = n^{O(1)}$ variables) in the form
\begin{equation}~\label{def:model}
P = \sum_{i=1}^T \alpha_i\cdot \prod_{j = 1}^d Q_{ij}
\end{equation}
where each $Q_{ij}$ is an arbitrary polynomial (of arbitrarily high degree) in at most $s$ variables and each $\alpha_i$ is a field constant. We call this the model of $\sop$ circuits. We use the quantity $Td$ as a measure of the size of a $\sop$ circuit.   Without loss of generality, we can assume that the degree zero term in each of the $Q_{ij}$ is either zero or one. If it is a non-zero constant other than $1$, we can extract it out and absorb it in $\alpha_i$. For each of the product gates, the fan-in could be different, but  we can assume without loss of generality that all the product fan-ins are equal to $d$. Observe that the $d$ could be much larger than the degree of the polynomial $P$. Throughout this paper, we will be working over a field of characteristic zero. 
\paragraph{Some basic notations : }
\begin{enumerate}
\item For an integer $i$, we denote the set $\{1, 2, \ldots, i\}$ by $[i]$.
\item By $\overline{X}$, we mean the set of variables $\{X_1, X_2, \ldots, X_N\}$. 
\item For a polynomial $P$ and a positive integer $i$, we represent by $\h^i[P]$, the homogeneous component of $P$ of degree equal to $i$. By $\h^{\leq i}[P]$ and $\h^{\geq i}[P]$, we represent the component of $P$ of degree at most $i$ and at least $i$ respectively. 
\item The support of a monomial $\alpha$ is the set of variables which appear with a non-zero exponent in $\alpha$. We denote the size of the support of $\alpha$ by $\text{Supp}(\alpha)$. 
\item Throughout the paper, we say that a function $f(N)$ is subexponential in $N$ if there exists a positive real number $\epsilon$, such that $\epsilon < 1$ and for all $N$ sufficiently large, $f(N) < \exp(N^{\epsilon})$. 
\item  We say that a function $f(N)$ is quasipolynomial in $N$ if there exists a positive absolute constant $c$, such that  for all $N$ sufficiently large, $f(N) < \exp(\log^c N)$. 
\item In this paper, we only consider layered arithmetic circuits and we will be counting levels from top to bottom, starting with the output gates being at level one. 
\item By a $\Sigma\Pi\Sigma\wedge$ circuit, we refer to a depth four circuit with all the product gates at the lowest level being replaced by powering ($\wedge$) gates. Similarly, by a $\Sigma\Pi\Sigma\wedge\Sigma\Pi$ circuit, we mean a depth six circuit all of whose product gates at level four from the top are powering gates. 
\end{enumerate}
\paragraph{Hitting set : } Let ${\cal C}$ be a set of polynomials in $N$ variables over a field $\F$. Then, a set  ${\cal H} \subseteq \F^{N}$ is said to be a {\it hitting set} for the class ${\cal C}$, if for every polynomial $P \in {\cal C}$ such that $P$ is not the identically zero polynomial, there exists a $p \in {\cal H}$ such that $P(p) \neq 0$.  

\paragraph{Elementary symmetric polynomials : } For variables $\overline{X} = \{X_1, X_2, \ldots, X_N\}$ and any integer $0 \leq l \leq N$,  the elementary symmetric polynomial of degree $l$ on variables $\overline{X}$ is defined as $$\es_l(\overline{X}) = \sum_{S \subseteq [N], |S| = l} \prod_{j \in S} X_j$$

\paragraph{Projected shifted partial derivatives :}
A key idea behind the recent progress on lower bounds is the notion of {\it shifted partial derivatives} introduced in~\cite{Kayal12}. In this paper, we use a variant of the measure, called projected shifted partial derivatives introduced in~\cite{KLSS14} and subsequently used in~\cite{KS-full}. Although we  never explicitly do any calculations with the measure in this paper, we provide a brief introduction to it below since the bounds are based on it.

For a polynomial $P$ and a monomial $\gamma$,  ${\partial_{\gamma} (P)}$ is the partial derivative of $P$ with respect to $\gamma$. For every polynomial $P$ and a set of monomials ${\cal M}$,  $\partial_{\cal M} (P)$ is the set of partial derivatives of $P$ with respect to monomials in ${\cal M}$. The space of $({\cal M}, m)\mhyphen$projected shifted partial derivatives of a polynomial $P$ is defined below. 
\begin{define}[$({\cal M}, m)\mhyphen$projected shifted partial derivatives]\label{def:shiftedderivative}
For an $N$ variate polynomial $P \in {\field{F}}[X_1, X_2, \ldots, X_{N}]$, set of monomials ${\cal M}$ and a positive integer $m\geq 0$, the space of $({\cal M}, m)$-projected shifted partial derivatives of $P$ is defined as
\begin{align}
 \langle \partial_{\cal M} (P)\rangle_{m} \stackrel{def}{=} \field{F}\mhyphen span\{\sigma(\prod_{i\in S}{X_i}\cdot g)  :  g \in \partial_{\cal M} (P), S\subseteq [N], |S| = m\}
\end{align}
\end{define}
Here, $\sigma(P)$ of a polynomial $P$ is the projection of $P$ on the multilinear monomials in its support. 
The measure of complexity of a polynomial that we use in this paper, is the dimension of projected shifted partial derivative space of $P$ with respect to some set of monomials ${\cal M}$ and a parameter $m$. Formally, 
$$\Phi_{{\cal M}, m} (P) = \dim( \langle \partial_{\cal M} (P)\rangle_{m})$$

From the definitions, it is straight forward to see that the measure is subadditive.
\begin{lem}[Sub-additivity]~\label{subadditive}
Let  $P$ and $Q$ be any two multivariate polynomials in $\F[X_1, X_2, \ldots, X_{N}]$.  Let ${\cal M}$ be any set of monomials and $m$ be any positive integer. Then, for all scalars $\alpha$ and $\beta$
$$\Phi_{{\cal M}, m} (\alpha\cdot P + \beta\cdot Q) \leq \Phi_{{\cal M}, m} (P) + \Phi_{{\cal M}, m} (Q)$$
\end{lem}

\paragraph{Approximations : }We will refer to the following lemma to approximate expressions during our calculations. 

\begin{lem}[\cite{GKKS12}]~\label{lem:approx}
Let $a(n), f(n), g(n) : \Z_{>0}\rightarrow \Z_{>0}$ be integer valued functions such that $(f+g) = o(a)$. Then,
$$\log \frac{(a+f)!}{(a-g)!} = (f+g)\log a \pm O\left( \frac{(f+g)^2}{a}\right)$$
\end{lem} 

In the proofs in this paper, we  use Lemma~\ref{lem:approx} only in situations where $(f+g)^2$ will  be $O(a)$. In this case, the error term will be bounded by an absolute constant. So, up to  constant factors, $\frac{(a+f)!}{(a-g)!} = a^{(f+g)}$. We use the symbol $\approx$ to indicate equality up to  constant factors.

\paragraph{Complexity of coefficients and homogeneous components :} We now summarise two simple lemmas which are useful for our proof.  The first lemma summarises that given a circuit $C$ for a polynomial $P \in \F[X1, X2, \ldots, X_{N}, Y]$ of degree at most $d$, for every $0 \leq i \leq d$, the coefficient of $Y^i$ in $P$ (when viewing $P$ as a polynomial in $\F[X_1, X_2, \ldots, X_{N}][Y]$) can also be computed by a circuit of size not much larger than the size of $C$.

\begin{lem}~\label{lem:extracting coefficients}
Let $P \in \F[X_1, X_2, \ldots, X_{N}, Y]$ be a polynomial of degree at most $d$ in $Y$ over a field  $\F$ of characteristic zero, such that $P$ is computable by an arithmetic circuit $C$ of size $|C|$. 
Let $$P = \sum_{i = 0}^d Q_i(X_1, X_2, \ldots, X_{N})\cdot Y^i$$
for polynomials $Q_i(X_1, X_2, \ldots, X_{N}) \in \F[X_1, X_2, \ldots, X_{N}]$.  
Then, for every $i$ such that $0 \leq i \leq d$, the polynomial $Q_i$ can be computed by an arithmetic circuit $C'$ of size at most $|C|\cdot (d+1)$. Moreover, if the output gate of $C$ is a $+$ gate, then the depth of $C'$ is equal to the depth of $C$. Else, the depth of $C'$ is at most $1$ more than the depth of $C$. 
\end{lem}
\begin{proof}
 We can view $P$ as a univariate polynomial of degree at most $d$ in $Y$ with the coefficients coming from $\F(\overline{X})$.  From the classical Lagrange interpolation, we know that the coefficient of $Y^i$ in $P$ can be written as an $\F(\overline{X})$ linear combination of the evaluations of $P$ at $d+1$ distinct values of $Y$ taken from $\F(\overline{X})$. In fact, more strongly, we can evaluate $P$ at $d+1$ values of $Y$ all chosen from $\F$ itself, in which case the constants in the linear combination are also from $\F$.  
So, $Q_i$ can be computed by a circuit obtained from taking $d+1$ circuits each obtained from $P$ by substituting $Y$ by a scalar in $\F$,  and taking their linear combination. Let this circuit be $C'$. Clearly the size of $C'$ is at most $(d+1)$ times the size of $C$. If the output gate of $C$ was an addition gate, then the outer addition for the linear combination can be absorbed into it, and the depth remains the same. Else, the depth increases by one. 
\end{proof}

The second  lemma stated below essentially says that the circuit complexity of homogeneous components of a polynomial is not much larger than the circuit complexity of the polynomial itself. 

\begin{lem}~\label{lem:interpolation}
Let $P$ be a polynomial of degree at most $d$ in $N$ variables over a field  $\F$ of characteristic zero, such that $P$ is computable by an arithmetic circuit $C$ of size $|C|$. Then, for every $i$ such that $0 \leq i \leq d$, the homogeneous component of degree $i$ of $P$ can be computed by an arithmetic circuit $C'$ of size at most $|C|\cdot (d+1)$. Moreover, if the output gate of $C$ is a $+$ gate, then the depth of $C'$ is equal to the depth of $C$. Else, the depth of $C'$ is at most $1$ more than the depth of $C$. 
\end{lem}
\begin{proof}
Let $P'(t)$ be the polynomial obtained from $P$ by replacing every variable $X$ in $P$ by $X\cdot t$ for a new variable $t$. We can view $P'$ to be a univariate polynomial of degree at most $d$ in $t$ with the coefficients coming from $\F(\overline{X})$. Observe that for every $i$ such that $0 \leq i \leq d$, the homogeneous component of $P$ of degree equal to $i$ is equal to the coefficient of $t^i$ in $P'$. The proof now follows from  Lemma~\ref{lem:extracting coefficients}. 
\end{proof}

\section{Proof of the lower bound}~\label{sec:lower bound}
In this section, we give the proof of Theorem~\ref{thm:mainthm intro}. We prove the lower bound for a variant of the well known  family of Nisan-Wigderson polynomials defined by Kayal and Saha~\cite{KayalSaha14}.  
\subsection{Target polynomials for the lower bound}
We now define the family of polynomials of degree $n$ in $N$ variables for which we prove the lower bounds. The family is a variant of the Nisan-Wigderson polynomials which were introduced by Kayal et al in~\cite{KSS13} in the context of lower bounds for homogeneous depth four circuits. The particular variant we use in the paper is due to Kayal and Saha~\cite{KayalSaha14}. 

The tradeoff between the number of variables $N$ and the degree $n$ will be parameterized by the parameter $\mu$ where $0 \leq \mu < 1$.  
First we need some parameters, which we define below. 
\begin{enumerate}
\item $\delta = (1-\mu)/2$ is a positive real number such that $\mu + \delta < 1$. 
\item $\gamma = \frac{2(\mu + \delta) + 1}{1-\mu-\delta} $. 
\item $N$ is chosen such that $N/n$ is a prime number between $n^{1 + \gamma}$ and $2n^{1+\gamma}$. Such a prime number always exists from the Bertrand-Chebychev theorem. Without loss of generality, we pick the smallest one. 
\item $\rho = (\mu + \delta)\frac{\log N}{\log n}$
\item $D   = \frac{\gamma + \rho}{2(1 + \gamma)} \cdot n$ , where $D-1$ is the degree of the underlying univariate polynomials in the definition of $\nwm$.
\end{enumerate}
Let $\psi$ be the prime number equalling $N/n$. We are now ready to restate the definition of $\nwm$ from~\cite{KayalSaha14}. 
\begin{define}[Nisan-Wigderson Polynomials~\cite{KayalSaha14}]~\label{defn:NW}  Let $\mu$ be a real number such that $0 \leq \mu < 1$.  For a given $\mu$ and $n$, let $N$, $D$, $\psi$ be as defined above. For the set of $N$ variables $\{X_{ij} : i\in [n], j \in [\psi]\} $, we define the degree $n$ homogeneous polynomial $\nwm$ as 
$$\nwm = \sum_{\substack{f(z) \in \F_{\psi}[z] \\
                        deg(f) \leq D-1}} \prod_{i \in [n]} X_{if(i)}$$
\end{define}                     
                       
From the definition, we can observe the following properties of $\nwm$. 
\begin{enumerate}
\item The number of monomials in $\nwm$ is exactly ${\psi}^{D} = n^{O(D)}$. 
\item Each of the monomials in $\nwm$ is multilinear.
\item Each monomial corresponds to evaluations of a univariate polynomial of degree at most $D-1$ at all points of $\F_{\psi}$. Thus, any two distinct monomials agree in at most $D-1$ variables in their support. 
\end{enumerate}
We will also need the following lemma in our proof. 
\begin{lem}~\label{lem: NW eval}
Let $\mu$ be a non-negative real number less than $1$.  Given $q \in \F^N$, $\mu$, $n$, we can evaluate the polynomial $\nwm$ at $q$ in time $N^{O(n)}$. 
\end{lem}
\begin{proof}
Given $n$ and $\mu$, we first find $D$, $\psi$ as given by the choice of parameters. 
Once we have $D$, we iterate through every monomial $\alpha$ of degree $n$ in the $\overline{X}$ variables which is supported on all the rows of the variable matrix and check if it is in the polynomial $\nwm$ by trying to find a univariate polynomial  $f(z) \in \F_{\psi}[z]$ such that degree of $f$ is at most $D-1$ and  $\prod_{i \in [n]} X_{if(i)} = \alpha$. The interpolation takes only $\text{Poly}(n)$ time, and the total number of monomials to try is at most $N^n$. So, we get the lemma.  
\end{proof}

We now proceed with the proof as outlined in Section~\ref{sec:overview lower bounds}. 
\subsection{Reducing the product fan-in at level two}
Let $P$ be a homogeneous polynomial in $N$ variables of degree $n$ which has a $\sop$ circuit of top fan-in $T$ and product fan-in $d$ at the second level. In other words, there exist polynomials $\{Q_{ij} : i \in [T], j \in [d]\}$ in at most $s$ variables each, such that 
\begin{equation}~\label{def:model2}
P = \sum_{i=1}^T \alpha_i\cdot\prod_{j = 1}^d Q_{ij}
\end{equation} 
Recall that without loss of generality, we can assume  that the constant term in each of the $Q_{ij}$ is either $0$ or $1$. We have the following lemma. 
\begin{lem}~\label{lem:homog}
Let $\F$ be a field of characteristic zero. Let $P$ be a homogeneous polynomial of degree $n$ in $N$ variables over $\F$ as defined above.  For each $i$,  $1\leq i \leq T$ define the set $$S_i = \{j : 1 \leq j \leq d \text{ and } \h^{0}[Q_{ij}] = 1\}$$ Then,
\begin{equation}
P = \sum_{i = 1}^T \alpha_i \cdot \h^n\left[\prod_{j \notin S_i} Q_{ij} \times \sum_{l = 0}^{n} \es_l(\{\h^{\geq 1}[Q_{ij}] : j \in S_i\})\right] 
\end{equation}
\end{lem}

\begin{proof}

To prove the lemma, we will try to extract out the homogeneous part of degree $n$ of each product gate $\prod_{j = 1}^d Q_{ij}$. Together with the fact that the polynomial $P$ is homogeneous of degree $n$, we get the lemma.  Every $Q_{ij}$ with a non-zero constant term can be written as $\h^{\geq 1}[Q_{ij}] + 1$, since the constant term in each $Q_{ij}$ is either $0$ or $1$. Now,  
\begin{equation}~\label{eqn:1}
\prod_{j = 1}^d Q_{ij} = \prod_{j \notin S_i} Q_{ij} \times \prod_{j \in S_i} (\h^{\geq 1}[Q_{ij}] + 1)
\end{equation}
  Decomposing the product $\prod_{j \in S_i} (\h^{\geq 1}[Q_{ij}] + 1)$ further, we have 
 \begin{equation}~\label{eqn:2}
 \prod_{j \in S_i} (\h^{\geq 1}[Q_{ij}] + 1] = \sum_{l = 0}^{|S_i|} \sum_{U \subseteq S_i : |U| = l} \prod_{j \in U} \h^{\geq 1}[Q_{ij}]
\end{equation}  
Now, observe that the degree of every monomial in $\prod_{j \in U} \h^{\geq 1}[Q_{ij}]$ is at least as large as the size of $U$. So, for every subset $U$ of size larger than $n$, $\prod_{j \in U} \h^{\geq 1}[Q_{ij}]$ is a polynomial of degree strictly larger than $n$. Also, for any fixed $l$, the expression $ \sum_{U \subseteq S_i : |U| = l} \prod_{j \in U} \h^{\geq 1}[Q_{ij}]$ is precisely the elementary symmetric polynomial of degree $l$ in the set of variables $\{\h^{\geq 1}[Q_{ij}] : j \in S_i\}$. Therefore,  
\begin{equation}~\label{eqn:3}
 \h^{\leq n}\left[\prod_{j \in S_i} (\h^{\geq 1}[Q_{ij}] + 1)\right] = \h^{\leq n}\left[\sum_{l = 0}^{n} \es_l(\{\h^{\geq 1}[Q_{ij}] : j \in S_i \})\right]
\end{equation}  
Therefore, 
\begin{equation}~\label{eqn:4}
\h^{n}\left[\prod_{j = 1}^d Q_{ij}\right] = \h^{n}\left[\prod_{j \notin S_i} Q_{ij} \times \sum_{l = 0}^{n} \es_l(\{\h^{\geq 1}[Q_{ij}] : j \in S_i \})\right]
\end{equation}
Summing up for all $i$, we get the lemma.
\end{proof}

The lemma above has in some sense helped us locate the monomials of degree $n$ in the circuit, which otherwise has a much higher formal degree. We now combine the above lemma with the well known fact that elementary symmetric polynomial of degree $l$ in $k$ variables can be computed by homogeneous $\Sigma\Pi\Sigma\wedge$ circuits of size at most $k2^{O(\sqrt{l})}$ to obtain a $\Sigma\Pi\Sigma\wedge\Sigma\Pi$ circut $C'$ such that the fan-in of the product gates at level two is at most $n$.   We use the following theorem (Theorem 5.2) by Shpilka and Wigderson~\cite{SW01}.
\begin{thm}[Shpilka-Wigderson~\cite{SW01}]~\label{thm : SW}
For every set of variables $\{Y_1, Y_2, \ldots, Y_m\}$ and a positive integer $l$, $\es_l(\{Y_1, Y_2, \ldots, Y_m\})$ can be computed by a homogeneous $\Sigma\Pi\Sigma\wedge$ circuit of size $m2^{O(\sqrt{l})}$. 
\end{thm}   

We now prove the following lemma. 
\begin{lem}~\label{lem:depth6}
Let $\F$ be a field of characteristic zero. Let $P$ be a polynomial of degree $n$ in $N$ variables over $\F$ which is computable by an $\sop$ circuit $C$ of top fan-in $T$ and the degree of product gates at level two being $d$. So, $P$ can be represented as  $$P = \sum_{i=1}^T \alpha_i\cdot\prod_{j = 1}^d Q_{ij}$$ 
Then, $P$ can be represented as the homogeneous component of degree $n$ of a polynomial computed  by a $\Sigma\Pi\Sigma\wedge\Sigma\Pi$ circuit $C''$ with the following properties : 
\begin{enumerate}
\item The inputs to the $\wedge$ gates are the polynomials $\{\h^{\geq 1}[Q_{ij}] : 1 \leq i \leq T, 1 \leq j \leq d\}$
\item The fan-in of the $\times$ gates at the second level from the top is at most $n$
\item The top fan-in of $C''$ is at most $Tdn2^{O(\sqrt{n})}$.
\end{enumerate}
\end{lem}
\begin{proof}
From Lemma~\ref{lem:homog}, we know that for the set $S_i$ defined as $$S_i = \{j : 1 \leq j \leq d \text{ and } \h^{0}[Q_{ij}] = 1\}$$ the polynomial $P$ can be written as
$$P = \sum_{i = 1}^T \alpha_i \cdot \h^n\left[\prod_{j \notin S_i} Q_{ij} \times \sum_{l = 0}^n \es_l(\{\h^{\geq 1}[Q_{ij}] : j \in S_i\})\right]$$

which is the same as 
$$P = \h^n\left[ \sum_{i = 1}^T \alpha_i \cdot \prod_{j \notin S_i} Q_{ij} \times \sum_{l = 0}^n \es_l(\{\h^{\geq 1}[Q_{ij}] : j \in S_i\})\right] $$
Observe that the polynomial $\prod_{j \notin S_i} Q_{ij}$ has degree at least $d-|S_i|$. We remark that if $d-|S_i|$ is larger than $n$, then such product gates do not contribute anything to the degree $n$ component of the polynomial and hence can be discarded without loss of generality; hence we assume $n-(d-|S_i|) > 0$. So, we could confine the inner sum from $l = 0$ to $l = n-(d-|S_i|)$, and still preserve the degree $n$ part of the polynomial, which is what we are interested in. 
From Theorem~\ref{thm : SW}, we know that for every $0 \leq l \leq n$, we can compute the polynomial $\es_l(\{\h^{\geq 1}[Q_{ij}] : j \in S_i\})$ by a $\Sigma\Pi\Sigma\wedge$ circuit of top fan-in at most $d \times 2^{O(\sqrt{l})}  $ which takes as input the polynomials $\{\h^{\geq 1}(Q_{ij}) : 1\leq j \leq d\}$.  From the homogeneity of the circuits given by Theorem~\ref{thm : SW}, it follows that the product gates at level two of these circuits have fan-in at most the degree of polynomial they compute, which is at most $n-(d-|S_i|)$. So, it follows that the polynomial 
$$\tilde{P} = \left( \sum_{i = 1}^T \alpha_i \cdot \prod_{j \notin S_i} Q_{ij} \times \sum_{l = 0}^{n-(d-|S_i|)} \es_l(\{\h^{\geq 1}[Q_{ij}] : j \in S_i\})\right)$$
 can be computed by a $\Sigma\Pi\Sigma\wedge\Sigma\Pi$ circuit, with top fan-in at most $Tdn\cdot 2^{O(\sqrt{n})}$, which satisfies the conditions in the lemma. 
\end{proof}
Finally, given the circuit $C''$ constructed above, we can construct a circuit which computes the polynomial $P$ as given by Lemma~\ref{lem:interpolation}. For this, observe that the monomials of degree strictly larger than $n$ in any of  the $Q_{ij}$ do not contribute to degree $n$ part of $\tilde{P}$. So, we can drop them, while still preserving the degree $n$ part of $\tilde{P}$. Therefore, the  degree of $\tilde{P}$ can be upper bounded by $n^2d$.  We can recover the degree $n$ part of $\tilde{P}$ by interpolation which blows up the top fan-in by a factor of at most $n^2d$. 

In this process, the fan-in of the product gates at level two remains unchanged. Strictly speaking,  inputs to the powering gate $\wedge$ at level four may no longer be the polynomials $\h^{\geq 1}[Q_{ij}]$, since in the process of interpolation, we replaced every variable $X_i$ by $X_i.t$ in $\tilde{P}$ and looked at the resulting polynomial $\tilde{P'}$ as a univariate polynomial in $t$ over the function field $\F(\overline{X})$. We then evaluated $\tilde{P'}$ at sufficiently many values of $t \in \F$ and then took their $\F$ linear combination. 
So, each of the polynomials $\h^{\geq 1}[Q_{ij}]$ gives rise to many other polynomials, one each for different values of $t$. We will call them the {\it siblings} of $\h^{\geq 1}[Q_{ij}]$. The key observation for our proof is that the set of variables in the siblings of $\h^{\geq 1}[Q_{ij}]$ is the same as the set of variables in $\h^{\geq 1}[Q_{ij}]$. From the lemma and the discussion above, we have the following corollary. 

\begin{cor}~\label{lem:depth6-cor}
Let $\F$ be a field of characteristic zero. Let $P$ be a polynomial of degree $n$ in $N$ variables over $\F$ which is computable by an $\sop$ circuit $C$ of top fan-in $T$ and the degree of product gates at level two being $d$. So, $P$ can be represented as  $$P = \sum_{i=1}^T \alpha_i\cdot\prod_{j = 1}^d Q_{ij}$$ 
Then, $P$ can be computed  by a $\Sigma\Pi\Sigma\wedge\Sigma\Pi$ circuit $C''$ with the following properties : 
\begin{enumerate}
\item The inputs to the $\wedge$ gates are the siblings of polynomials $\{\h^{\geq 1}[Q_{ij}] : 1 \leq i \leq T, 1 \leq j \leq d\}$
\item The fan-in of the $\times$ gates at the second level from the top is at most $n$
\item The top fan-in of $C''$ is at most $Td^2n^32^{O(\sqrt{n})}$.
\end{enumerate}
\end{cor}

\subsection{Random Restrictions}~\label{sec: random res}
From the definition, it follows that the total number of variables in $NW_{n,\mu}$ is $N$.  Let the set of all these variables be $\cal V$.  We  now define our random restriction procedure by defining a distribution $\cal D$ over subsets $V \subset \cal V$. The random restriction procedure will sample $V \gets \cal D$ and then keep only those variables ``alive" that come from $V$ and set the rest to zero. We will denote the restriction of the polynomial obtained by such a restriction as $NW_{n, \mu}|_V$. Observe that a random restriction also results in a distribution over all circuits computing the polynomial $NW_{n, \mu}$. We denote by $C|_V$  the restriction of a circuit $C$ obtained by setting every input gate in $C$ which is labelled by a variable outside $V$ to $0$.

\vspace{2mm}
\noindent
{\bf The distribution ${\cal D}_p$: } Each variable in $\cal V$ is independently kept alive with a probability $p$. We will choose the value of $p$ based on the parameter $\mu$. 

\subsection{Analysing the circuit under random restrictions}
Let $C$ be a $\sopm$ circuit computing the polynomial $\nwm$. Let the top fan-in of $C$ be $T$ and the product fan-in at the second level be $d$. So, we have the following expression. 
$$\nwm = \sum_{i=1}^T \alpha_i\cdot\prod_{j = 1}^d Q_{ij}$$ 
where each $Q_{ij}$ depends on at most $N^{\mu}$ variables. 

Recall that from the choice of parameters $\delta = (1-\mu)/2$. Let $s$ be a parameter, which we later  set such that  $s = \Theta(\sqrt{n})$. If $T\cdot d \geq N^{\frac{\delta}{4} s}$, then we already have the desired lower bound of $n^{\Omega(\sqrt{n})}$ on the size of $C$ and we are done. Therefore, for the rest of this discussion, we will assume that $T\cdot d \leq N^{\frac{\delta}{4}s}$. We now apply the transformation to $C$ given by Corollary~\ref{lem:depth6-cor} to obtain a $\Sigma\Pi\Sigma\wedge\Sigma\Pi$ circuit $C''$, which has the following properties: 
\begin{enumerate}
\item The inputs to the $\wedge$ gates are the siblings of polynomials $\{\h^{\geq 1}[Q_{ij}] : 1 \leq i \leq T, 1 \leq j \leq d\}$
\item The fan-in of the $\times$ gates at the second level from the top is at most $n$
\item The top fan-in of $C''$ is at most $Td^2n^32^{O(\sqrt{n})}$.
\end{enumerate}
We now analyse the effect of the random restrictions on the circuit $C''$. We will choose a parameter $p = N^{-\mu-\delta}$ and keep every variable alive with a probability $p$. The circuit $C''$ can be represented as $$C'' = \sum_{u}\prod_{v} D_{uv}$$
Here, each $D_{uv}$ is a sum of powers of the siblings of  $\h^{\geq 1}[Q_{ij}]$. Our goal is to argue that under random restrictions, all the monomials in each of the $D_{uv}$ are of small support (support at most $s$). 

For any polynomial $P$ in $N^{\mu}$ variables and any integers $t, t_0$ such that $t_0 < t$, observe that $P^t$ can be written as 

$$P^t =  P_0 + \sum_{\alpha}\alpha\cdot P_{\alpha}$$
where $P_0$ is the part of $P$ consisting of monomials of support strictly less than $t_0$. The inner sum is over all multilinear monomials $\alpha$ of support equal to $t_0$. Such a decomposition may not be unique, but for this application, it would suffice to work with any one such decomposition. The number of such monomials $\alpha$ is at most  ${N^{\mu} \choose t_{0}}$. The probability that one such monomial survives the random restriction procedure is equal to $p^{t_0}$. So, the expected number of such multilinear monomials $\alpha$ surviving the random restriction procedure is at most ${N^{\mu} \choose t_{0}}\cdot p^{t_0}$. The crucial observation is that if no such monomials survive, then only the monomials in $P_0$ survive, all of  which have support at most $t_0-1$. 

Now, observe that each of the $D_{uv}$ are a sum of powers of the siblings of polynomials in the set $\{\h^{\geq 1}[Q_{ij}] : 1 \leq i \leq T, 1 \leq j \leq d\}$. Define ${\cal B}$ to be the set of all multilinear monomials of support equal to $s$, supported entirely on variables in any of the polynomials ${Q_{ij}}$ for some $1 \leq i \leq T, 1 \leq j \leq d$.  From the discussion in the paragraph above, the following observation follows. 

\begin{obs}~\label{obs: random rest 1}
Let the polynomials $D_{uv}$, $Q_{ij}$ and the set ${\cal B}$ be as defined above. Then, 
\begin{itemize}
\item $|{\cal B}| \leq T\cdot d\cdot {N^{\mu} \choose s}$
\item If none of the monomials in ${\cal B}$ survive under some random restrictions, then  each of the polynomials $D_{uv}'$ obtained as a restriction of $D_{uv}$ has all monomials of support at most $s$.  
\end{itemize}
\end{obs}
\begin{proof}
The bound on the size trivially follows from the fact that each of the $Q_{ij}$ depends on at most $N^{\mu}$ variables. For the second item, observe that each of the $D_{uv}$ is a sum of powers of siblings of  the $\h^{\geq 1}[Q_{ij}]$ and all the siblings are supported on the same set of variables. If all the monomials in the set ${\cal B}$ are set to zero, then the surviving monomials in any power of any of the siblings of  $\h^{\geq 1}[Q_{ij}]$ has support at most $s$. 
\end{proof}

We now estimate the probability that at least one of the monomials in the set ${\cal B}$ survives the random restriction procedure. We have the following lemma. 

\begin{lem}~\label{lem: rand res 2}
Let $\delta$ be a positive real number such that $\delta = (1-\mu)/2$ and let $p =  N^{-\mu-\delta}$. Then 
$$Pr_{V\leftarrow {\cal D}_p}\left[|{\cal B}|_{V}| \geq 1\right] \leq N^{-3/4 \cdot\delta \cdot s}$$  
\end{lem}
\begin{proof}
We know that $$|{\cal B}| \leq T \cdot d \cdot {N^{\mu} \choose s}$$ and the probability that any fixed monomial in ${\cal B}$ survives the random restriction procedure is at most $p^{s}$. So $$\e_{V\leftarrow {\cal D}_p}[|{\cal B}_{V}|] \leq T \cdot d \cdot {N^{\mu} \choose s} \cdot p^s  $$
Now, observing that the value of $T\cdot d$ is at most $N^{\frac{\delta}{4}s}$ and $p =  N^{-\mu-\delta}$, the expected value is at most $$ N^{\frac{\delta}{4}s} {N^{\mu} \choose s} \cdot  N^{-(\mu+\delta)s} \leq N^{-3/4 \cdot\delta \cdot s}$$
The lemma then follows by Markov's inequality. 
\end{proof}

As a corollary of Lemma~\ref{lem: rand res 2} and Observation~\ref{obs: random rest 1}, we get the following lemma.
\begin{lem}~\label{lem: rand res main}
Let $\delta$ be a positive real number such that $ \delta = (1-\mu)/2$ and let $p =  N^{-\mu-\delta}$. Then with probability at least $1- N^{-3/4 \cdot\delta \cdot s}$ over random restrictions $V \leftarrow {\cal D}_p$, the polynomial computed by the circuit $C''|_{V}$ can be written as $\sum_{u = 1}^{T'} \prod_{v = 1}^n D_{uv}'$, where each of the monomials in each of the polynomials $D_{uv}'$ has support at most $s$.  
\end{lem} 

\subsection{Upper bound on the complexity of C}
In order to upper bound the dimension of the projected shifted partial derivatives (under random restrictions) of the $\sop$ circuit $C$, Corollary~\ref{lem:depth6-cor} implies that it suffices to upper bound the dimension of the space of projected shifted partial derivatives of the $\Sigma\Pi\Sigma\wedge\Sigma\Pi$ circuit $C''$ given by Corollary~\ref{lem:depth6-cor}. In some sense, $C''$ is more structured than $C$ and this lets us prove a better upper bound. 

Recall that we are under the assumption that for the circuit  $C$, the product of the top fan-in and the product fan-in at level two is at most $N^{\frac{\delta}{4} \cdot s}$, else we are already done. 
From Lemma~\ref{lem: rand res main}, we know that with a high probability, under random restrictions, we are left with a circuit of the form $\sum_{u = 1}^{T'} \prod_{v = 1}^n D_{uv}'$ where each of the monomials in each of the polynomials $D_{uv}'$ has support at most $s$. The upper bound on the complexity of the projected shifted partial derivatives of $\sum_{u = 1}^{T'} \prod_{v = 1}^n D_{uv}'$ then just follows from the upper bound for homogeneous depth four circuits of bounded bottom support proved in~\cite{KLSS14, KS-full}. We restate the bound from~\cite{KS-full}.   

\begin{lem}~\label{lem:lowsupbound1}
Let $C$ be a depth 4 circuit with the fan-in or product gates at level two bounded by $n$, the bottom support bounded by $s$ and computing a polynomial in $N$ variables. Let ${\cal M}$ be a set of monomials of degree equal to $r$ and let $m$ be a positive integer. Then,  $$\Phi_{{\cal M}, m}(C) \leq \text{Top fan-in}(C){n + r \choose r}{N \choose m+ rs}$$ for any choice of $m, r, s, N$ satisfying $m+rs \leq N/2$.
\end{lem}

The upper bound for $\sopm$ circuits, follows easily form the above lemma after random restrictions, and we formalize this in the lemma below.

\begin{lem}~\label{lem:complexity ub}
Let $\mu$ be a positive real number such that $0 \leq \mu < 1$. Let $\delta = (1-\mu)/2$ and let $p =  N^{-\mu-\delta}$ and let $\F$ be a field of characteristic zero. Let $P$ be a polynomial of degree $n$ in $N$ variables over $\F$ which is  computed by an $\sopm$ circuit $C$ of top fan-in $T$ and degree of product gates at level two at most $d$, i.e $P$ can represented as  $$P = \sum_{i=1}^T \alpha_i\cdot\prod_{j = 1}^d Q_{ij}$$ where $\alpha_i$ are field constants.
Let $m$ and $r$ be  positive integers satisfying $m+rs \leq N/2$ and ${\cal M}$ be any subset of multilinear monomials of degree equal to $r$. 
If $Td \leq N^{\frac{s\cdot \delta}{4}}$, then with probability at least  $1- N^{-3/4 \cdot\delta \cdot s}$ over random restrictions $V \leftarrow {\cal D}_p$,  $$\Phi_{{\cal M}, m} (C|_V) \leq Td^2n^3 \cdot rs \cdot 2^{O(\sqrt{n})}\cdot {N \choose m+rs} \cdot {n + r \choose r} $$
\end{lem}
\begin{proof}
The lemma follows immediately from Corollary~\ref{lem:depth6-cor}, Lemma~\ref{lem: rand res main} and Lemma~\ref{lem:lowsupbound1}.
\end{proof}

\subsection{Nisan-Wigderson polynomial under random restrictions}
To complete the proof of Theorem~\ref{thm:mainthm intro}, we need a lower bound on the dimension of the space of projected shifted partial derivatives of the polynomial $\nwm$, under random restrictions. 
To this end, we will use the lower bound proved by Kayal and Saha~\cite{KayalSaha14}. 
We first enumerate our choice of parameters. Recall that $\delta = (1-\mu)/2$ is a positive real number.
\begin{enumerate}
\item $\gamma = \frac{2(\mu + \delta) + 1}{1-\mu-\delta}$ 
\item $N$ is such that $N/n$ is set equal to the smallest prime number between $n^{1 + \gamma}$ and $2n^{1+\gamma}$.  
\item $\rho = (\mu + \delta)\frac{\log N}{\log n}$
\item $D   = \frac{\gamma + \rho}{2(1 + \gamma)} \cdot n$ , where $D-1$ is the degree of the underlying univariate polynomials in the definition of $\nwm$.
\item $r, s$ which are the order of derivative and the bound on bottom support of the circuit after random restrictions respectively, are chosen such that $r = \epsilon_1\cdot \sqrt{n}, s = \epsilon_2\cdot \sqrt{n}$. Here, $\epsilon_1$ and $\epsilon_2$ are small enough positive real numbers  satisfying $\epsilon_1\cdot\epsilon_2 = 0.001 n $.  
\item $m = \frac{N}{2}(1-r\frac{\ln n}{n})$ is the degree of the shifts.
\item $p = N^{-(\mu + \delta)}$ is the probability with which each variable is independently kept alive. 
\item ${\cal M}$ is the set of all multilinear monomials of degree $r$. We take partial derivatives with respect to monomials in this set. 
\end{enumerate}
We are now ready to state the lower bound on the dimension of projected shifted partial derivatives as in~\cite{KayalSaha14}. 

\begin{lem}[Kayal-Saha~\cite{KayalSaha14}]~\label{lem: KS main}
Let $\nwm$ be Nisan-Wigderson polynomials as defined in Definition~\ref{defn:NW}. Let $\F$ be any field of characteristic zero. Then, for the choice of parameters defined above
$$\Phi_{{\cal M}, m}(\nwm|_V) \geq \frac{1}{n^{O(1)}}\text{min}\left(\frac{p^r}{4^r} \cdot {N \choose r} \cdot {N \choose m}, {N \choose m + n - r}\right) $$
with probability at least $1 - \frac{1}{n^{\theta(1)}}$ over random restrictions $V \leftarrow {\cal D}_p$. 
\end{lem}

\subsection{Wrapping up the proof of Theorem~\ref{thm:mainthm intro}}
From Lemma~\ref{lem: KS main} and Lemma~\ref{lem: rand res main}, we know that with a non-zero probability over the random restrictions $V$ from the distribution ${\cal D}_p$, the following two conditions hold. 
\begin{enumerate}
\item $$\Phi_{{\cal M}, m}(\nwm|_V) \geq \frac{1}{n^{O(1)}}\text{min}\left(\frac{p^r}{4^r} \cdot {N \choose r} \cdot {N \choose m}, {N \choose m + n - r}\right) $$
\item $$\Phi_{{\cal M}, m} (C|_V) \leq Td^2n^3 \cdot rs \cdot 2^{O(\sqrt{n})}\cdot {N \choose m+rs} \cdot {n + r \choose r}$$ 
\end{enumerate}
If $C$ computed the polynomial $\nwm$, then 
$$Td^2n^3 \cdot rs \geq \frac{{\frac{1}{n^{O(1)}}\text{min}\left(\frac{p^r}{4^r} \cdot {N \choose r} \cdot {N \choose m}, {N \choose m + n - r}\right)}}{{2^{O(\sqrt{n})}\cdot {N \choose m+rs} \cdot {n + r \choose r}}} $$

From the calculations in Appendix~\ref{sec:calc}, it follows that for our choice of parameters, the ratio is at least $\exp(\sqrt{n}\log n)$. So, we have the following theorem.

\begin{thm}~\label{thm:mainthm}
Let $\mu$ be an absolute constant such that $0 \geq \mu < 1$ and $\F$ be a field of characteristic zero. For $1 \leq i \leq T$ and $1 \leq j \leq d$, if there exist polynomials $Q_{ij}$, each dependent on only $s = N^{\mu}$ variables, such that 
$$\nwm = \sum_{i = 1}^T\prod_{j = 1}^{d} Q_{ij}$$
Then
$$T\cdot d \geq n^{\Omega_{\mu}(\sqrt{n})}$$
\end{thm}

As a remark, we mention here that the lower bound above also holds for any translation $\nwm(\overline{X} + \overline{a})$ of the polynomial $\nwm(\overline{X})$. This is because the highest degree term of  $\nwm(\overline{X} + \overline{a})$ equals the polynomial $\nwm(\overline{X})$ and from Lemma~\ref{lem:interpolation}, the homogeneous components of a polynomial computable by small sized $\sop$ circuits also have small sized $\sop$ circuits. We leave the details to the interested reader.  
\section{Application to polynomial identity testing}~\label{sec:pit}
In this section, we prove Theorem~\ref{thm:mainthm2 intro}. 
We are interested in identity testing for $\sop$ circuits, i.e for polynomials  in $N$ variables $\{X_1, X_2, \ldots, X_N\}$ which can be expressed in the form 
$$P = \sum_{i = 1}^T \prod_{j = 1}^d Q_{ij}$$ such that 
\begin{enumerate}
\item The individual degree in $P$ of every variable is at most $k$
\item Each $Q_{ij}$ depends on at most $s$ variables
\end{enumerate}
For the case of this application, we will think of $k, T$ being polynomial in $(\log N)$ and $s$ being $N^{1/2-\epsilon}$ for a positive constant $\epsilon$. Observe that the bound on individual degree lets us upper bound the total degree of the polynomials by $Nk$.

We describe the construction of the hitting set in Section~\ref{sec:hitting set} and prove its correctness in Section~\ref{sec:hitting set correct}. We go over some preliminaries that we need in our proof in the next section. 

\subsection{Some preliminaries} 
In the following lemma, we prove some properties of the model of $\sop$ circuits, which will be useful in the proof of the identity testing result. 
\begin{lem}~\label{lem: model props}
Let $\F$ be a field of characteristic zero. 
Let $P$ be a non-zero polynomial in $N$ variables and individual degree at most $k$ over $\F$, which is computed by a $\sop$ circuit C of top fan-in $T$ and product fan-in  $d$ at level two, i.e $P$ can be expressed as 
$$P = \sum_{i = 1}^T \prod_{j = 1}^d Q_{ij}$$
such that for each $i \in [T]$ and $j \in [d]$, $Q_{ij}$ depends on at most $s$ variables. Then, the following are true. 
\begin{enumerate}
\item For every variable $y$ and integer $1 \leq j \leq k$, $\frac{{\partial}^j P}{{\partial y^j}}$ can be computed by a circuit of the form $$\frac{\partial^j P}{\partial y^j} = \sum_{i = 1}^{T'} \prod_{j = 1}^d Q_{ij}'$$
where $T' \leq T\cdot (k+1)^2$ and each of the polynomials $Q_{ij}'$ depends on at most $s$ variables. 
\item For any $a \in \F^N$, $P(\overline{X} + \overline{a})$ can be computed by a circuit of the form $$P(\overline{X} + \overline{a}) = \sum_{i = 1}^{T} \prod_{j = 1}^d Q_{ij}''$$
where each of the polynomials $Q_{ij}''$ depends on at most $s$ variables.
\end{enumerate}
\end{lem} 
\begin{proof}
The proof of the second item is immediate from the definitions. The only thing that changes due to a translation is the number of monomials in the $Q_{ij}$. The number of variables that each $Q_{ij}$ depends on remains unchanged, and so does the fan-in of the top sum gate and the product gates at level two. 

We now prove the first item. Let the set of variables in $P$ be $\overline{X} = \overline{X'} \cup \{y\}$ where $X'$ is of size $N-1$. Since the individual degree of $P$ is at most $k$, we can write $P = \sum_{i = 0}^k C_i(\overline{X'})\cdot y^i$. Here, $C_i(\overline{X'})$ are polynomials only in the $X'$ variables and are the coefficient of $y^i$, when viewing $P$ as an element of $\F[\overline{X'}][y]$. Now, for every $0 \leq i \leq k$, we can compute each of $C_i$ by a $\sop$ circuit with top fan-in at most $T\cdot(k+1)$ by interpolation as given by Lemma~\ref{lem:extracting coefficients}. All the partial derivatives of $P$ with respect to $y$ are linear combinations of the terms of the form $C_{j_1}\cdot y^{j_2}$. And so, the result follows. 
\end{proof}

We will also need the following simple fact about polynomials. 
\begin{lem}~\label{lem:non zero derivative}
Let $\F$ be a field of characteristic zero. 
Let $R \in \F[y]$ be a non-zero polynomial of degree at most  $t$ over the field $\F$. Then, for every $a \in \F$ such that $R(a) = 0$, there exists a $j$ such that $0 \leq j \leq t-1$ and  $\frac{\partial^j R}{\partial y^j}(a) = 0$ and $\frac{\partial^{j+1} R}{\partial y^{j+1}}(a) \neq 0$. 
\end{lem}
\begin{proof}
Let the degree of $R$ in $y$ be equal to $t'$. This means that the coefficient of highest degree term $y^{t'}$ in $R$ is non-zero. Let us call the coefficient of $y^{t'}$ in $R(y)$ as $C_{t'}$. We know that $C_{t'}$ is nonzero.  Consider $j = t'-1$. The lemma immediately follows. 
\end{proof}

We will crucially use the following result of Dvir, Shpilka, Yehudayoff~\cite{DSY09} in the analysis of the hitting set constructed in this paper. 
\begin{lem}[Dvir, Shpilka, Yehudayoff~\cite{DSY09}]~\label{lem:DSY main}
For a field $\F$, let $P \in \F[X_1, X_2, \ldots, X_N, Y ]$ be a non-zero polynomial of degree at most $k$ in $Y$. Let $f \in \F[X_1, X_2, \ldots, X_N]$ be a polynomial such that $P(X_1, X_2, \ldots, X_N, f) = 0$ and $\frac{\partial P}{\partial Y} (0, 0, \ldots, 0, f(0, 0, \ldots, 0))\neq 0$. Let $$P = \sum_{i = 0}^k C_i(X_1, X_2, \ldots, X_N)\cdot y^i$$ Then, for every $t \geq 0$, there exists a polynomial $R_t \in \F[Z_1, Z_2, \ldots, Z_{k+1} ]$ of degree at most $t$ such that $$\h^{\leq t}[f(X_1, X_2, \ldots, X_N)] = \h^{\leq t}[R_t(C_0, C_1, \ldots, C_k)] $$
\end{lem} 
A key technical idea in the proof will be the notion of Nisan-Wigderson designs introduced in~\cite{NW94}. We will use the following lemma. 
\begin{lem}[Nisan-Wigderson~\cite{NW94}]~\label{lem: designs}
For every $a, b \in \N$, $b < 2^a$, there exists a family of sets $S_1, S_2, \ldots, S_b \subseteq \{1, 2, \ldots, l\}$ such that 
\begin{enumerate}
\item $l \in O(a^2/\log b)$ 
\item for all $i$, $|S_i| = a$
\item for all $i \neq j$, $|S_i \cap S_j| \leq \log b$
\end{enumerate}
Moreover, such a set family can be constructed in time polynomial in $b$ and $2^l$.  
\end{lem}
We will also use the following lemma of Alon~\cite{AlonCN} very crucially in our proof.
\begin{lem}[Combinatorial Nullstellensatz~\cite{AlonCN}]~\label{lem: comb nulls}
Let $P$ be a non-zero polynomial of  individual degree at most $d$ in $N$ variables over a large enough field $\F$. Let $S$ be an arbitrary subset of $\F$ of size $d+1$. Then, there exists a point $p$ in $S^{N}$ such that $P(p) \neq 0$.  
\end{lem}

\subsection{Blackbox PIT for $\sop$ circuits}~\label{sec:hitting set}
In this section, we prove the following theorem. 
\begin{thm}~\label{thm:mainthm2}
Let $c$ and $\mu$ be arbitrary constants such that $c> 0$ and  $0 \leq \mu < 1/2$, and let $\F$ be a field of characteristic zero. Let ${\cal C}$ be the set of polynomials $P$ in $N$ variables and individual degree at most $k$ over $\F$, with the property that $P$ can be expressed as 
$$P = \sum_{i = 1}^T \prod_{j = 1}^d Q_{ij}$$
such that 
\begin{enumerate}
\item $T < \log^c N$
\item $k < \log ^c N$
\item $d < N^c$
\item each $Q_{ij}$ depends on at most $N^{\mu}$ variables 
\end{enumerate}
Then, there exists a constant $\epsilon < 1$ dependent only on $c$ and $\mu$, such that there is a hitting set of size $\exp(N^{\epsilon})$ for ${\cal C}$ which can be constructed in time $\exp(N^{\epsilon})$. 
\end{thm}
From our proof, it also follows that if each of polynomial $Q_{ij}$  depends only on $\log^{O(1)} N$ variables, then both the size of the hitting set and the time to construct it, are upper bounded by a quasipolynomial function in $N$. 
 In the rest of the section, we prove Theorem~\ref{thm:mainthm2}. We start by describing the construction of the hitting set $\cal H$. 

\subsubsection{Construction of  hitting sets for $\sopm$ circuits for $0 \leq \mu < 1/2$} 
Given $\mu$ such that $0 \leq \mu < 1/2$, we pick the parameter $\mu'$ such that $0 < \mu' < 1$ and $\frac{2\mu}{\mu'}$ is a positive constant strictly smaller than $1$. 
We construct a family of Nisan-Wigderson designs as described in Lemma~\ref{lem: designs} with the following parameters :
\begin{enumerate}

\item $b$, the number of sets is set equal to $N$
\item $a$, the size of each of the sets $S_i$ is set equal to $N^{\frac{\mu}{\mu'}}\log^{\frac{1}{\mu'}} N$.
\item $l$, the size of the universe  is chosen large enough in order to satisfy the hypothesis of Lemma~\ref{lem: designs}. From 
Lemma~\ref{lem: designs}, it follows that we can pick $l$ which is not too large ($l \in O(a^2/\log b)$). For the above chosen values of $a, b$,  there is a choice of $l$ such that $l$ is at most  $N^{\frac{2\mu}{\mu'}}\log^{\frac{2}{\mu'}-1} N$.

\end{enumerate}
Recall that our goal is to construct a hitting set for $\sopm$ circuits. Observe that the choice of parameters $l, a, b$ satisfy the hypothesis of Lemma~\ref{lem: designs}. 
So, we get a collection of $N$ subsets $S_1, S_2, \ldots, S_N$ of $\{1, 2, 3, \ldots, l\}$ satisfying 
\begin{enumerate}
\item for all $1\leq i \leq N$,  $|S_i| = a$
\item for all $1 \leq i < j \leq N$, $|S_i \cap S_j| \leq \log N$
\end{enumerate}
Moreover, these sets can be constructed in time polynomial in $b$ and $2^l$. 
We identify the set $\{1, 2, 3, \ldots, l\}$ with the set of new variables $\overline{Y} = \{Y_1, Y_2, \ldots, Y_l\}$.
Before we proceed further, we need some notation. We will pick $\delta = (1-\mu')/2$ to be a non-negative constant. Given, $a, \mu', \delta$, we define $\gamma = \frac{2(\mu' + \delta) + 1}{1-(\mu' + \delta)}$. Then, we define $q$ to be the smallest prime number between $({a/2})^{\frac{1+\gamma}{2+\gamma}}$ and $2\cdot ({a/2})^{\frac{1+\gamma}{2+\gamma}}$. Also, we set $a'$ to be equal to  $({a/2})^{\frac{1}{2+\gamma}}$. 
Observe that $a/2 \leq a'q \leq a$. 

For each $i$, such that $1 \leq i \leq N$, let ${S_i}'$ be an arbitrary subset of $S_i$ of size equal to $a'q$. For brevity, we rename the sets $S_i'$ as $S_i$~\footnote{We have replaced the family $\{S_1, S_2, \ldots, S_N\}$ by the set family $\{S_1', S_2', \ldots, S_N'\}$ such that for each $i \in [N]$, $S_i' \subseteq S_i$. Observe that the design based properties of the original system continue to hold. The only thing that changes is that the size of $S_i'$ could be smaller than the size of $S_i$, by at most a factor $2$.  }. Let $\rho = (\mu' + \delta)\frac{\log a'q}{\log a'}$ and $D   = \frac{\gamma + \rho}{2(1 + \gamma)} \cdot a'$.

 Often for the ease of notation we will identify the set $S_i$ of $\{1, 2, \ldots, l\}$ with the set of variables $\{Y_j : j \in S_i\}$. We will think of the variables $\{Y_j : j \in S_i\}$ to be arranged in a $a'\times q$ matrix $V(i)$, with the variables placed in the matrix in some order. 
For every $i\in \{1, 2, 3, \ldots, N\}$, we define $\nwa(S_i)$ as 
$$\nwa(S_i) = \sum_{\substack{f(z) \in \F_{q}[z] \\
                        deg(f) \leq D-1}} \prod_{j \in [a']} V(i)_{jf(j)}$$

For a point $p = (p_1, p_2, \ldots, p_l) \in \F^l$, we denote by $\nwa(S_i)|p$, the evaluation of $\nwa(S_i)$ when the variable $Y_j$ is set to $p_j$.

Let $G$ be  an arbitrary subset of $\F$ of size $Nka' + 1$. We define the hitting set ${\cal H}$ as follows. 

\begin{define}[Definition of the hitting set ${\cal H}$]~\label{def:hitting set}
$${\cal H} = \left\{ (\nwa(S_1)|p, \nwa(S_2)|p, \ldots, \nwa(S_N)|p) : p \in G^{l} \right\}   $$
\end{define}

We now proceed to prove the correctness of the construction. We first prove the following lemma which shows that ${\cal H}$ is explicit and has the correct size as per Theorem~\ref{thm:mainthm2}. 

\begin{lem}~\label{lem: hitting set size}
The set ${\cal H}$ as defined in Definition~\ref{def:hitting set} has size at most $(Nka' + 1)^l$ and all its elements can be enumerated in time $a^{a'}\cdot (Nka' + 1)^l\cdot N^{O(1)}$. 
\end{lem}

\begin{proof}
The size of the set ${\cal H}$ is equal to $|G|^l = (Nka' + 1)^l$.  The set ${\cal H}$ can be  enumerated by enumerating through the points $p$ in $G^l$ in some natural order (say lexicographic order) and evaluating the tuple $ (\nwa(S_1)|p, \nwa(S_2)|p, \ldots, \nwa(S_N)|p)$ at each of these points. For every point $p$ and subset $S_i$, the polynomial $\nwa(S_i)$ can be evaluated in time at most $a^{a'}\times \text{Poly}(N)$ from Lemma~\ref{lem: NW eval}.  So, the second part of the lemma follows.  
\end{proof}

Observe that for our choice of parameters, the above bounds on the size and the time of enumeration are bounded by a function which is subexponential in $N$. 

We now show that for every non-zero polynomial $P$ in the class ${\cal C}$, as defined in the statement of Theorem~\ref{thm:mainthm2}, there exists a point $p \in {\cal H}$, such that $P(p)$ is non-zero. We show this in Lemma~\ref{lem: hitting set correctness} below. That will complete the proof of Theorem~\ref{thm:mainthm2}.  


\subsection{Correctness of the construction}~\label{sec:hitting set correct}
For the rest of this section, we denote $N^{\mu}$ by $s$.
\begin{lem}~\label{lem: hitting set correctness}
Let $P$ be a non-zero polynomial in the set $\cal C$ as defined in the statement of Theorem~\ref{thm:mainthm2}, and let ${\cal H}$ be the set defined in Definition~\ref{def:hitting set}.  Then, there is a point $p$ in the set ${\cal H}$ such that $P(p) \neq 0$.
\end{lem}
\begin{proof}
We define $$P_i(\overline{X}, \overline{Y}) := P(\nwa(S_1), \nwa(S_2), \ldots, \nwa(S_i), X_{i+1}, X_{i+2}, \ldots, X_N)$$ to be the polynomial obtained from $P$ by substituting the variables $X_j$ by $\nwa(S_j)$, for every $1 \leq j \leq i$. 

From the construction of our hitting set, it follows that it would suffice to argue that the polynomial $P_{N}(\overline{X}, \overline{Y})$ is non-zero. If this was true, then the lemma above will follow from Lemma~\ref{lem: comb nulls}, since the degree of any variable $P(\overline{X}, \overline{Y})$ is at most $Nka'$. 

 We proceed via contradiction. If possible, let $P_N(\overline{X}, \overline{Y})$ be identically zero. Since $P = P_0(\overline{X}, \overline{Y})$ is non-zero to start with, by a hybrid argument it follows that there is an index $i$, such that $P_i(\overline{X}, \overline{Y})$ is non-zero while $P_{i+1}(\overline{X}, \overline{Y})$ is identically zero. Observe that $P_i$ is a polynomial in the variables $\overline{Y}$ and $X_{i+1}, X_{i+2}, \ldots, X_N$. 
In going from $P_{i}$ to $P_{i+1}$, we substituted the variable $X_{i+1}$ by the polynomial $\nwa(S_{i+1})$. Since $P_{i}(\overline{X}, \overline{Y})$ is  non-zero by assumption above, there exists a substitution $\overline{c}$ of all variables apart from $\{Y_j : j \in S_{i+1}\}$ and $X_{i+1}$, which keeps the polynomial non-zero. Let the polynomial resulting after this substitution be $P_i'$. From the definitions, it follows that 
$$P_i' = P(\nwa(S_1)|{\overline{c}}, \nwa(S_2)|{\overline{c}}, \ldots, \nwa(S_i)|{\overline{c}}, X_{i+1}, X_{i+2}|{\overline{c}}, \ldots, X_N|{\overline{c}}) $$ 

Observe that each of the polynomials $\nwa(S_j)|{\overline{c}}$ depends only on the variables in the set $S_j \cap S_{i+1}$. From the properties of Nisan-Wigderson designs, and the choice of parameters, the size of this intersection is at most $\log N$. From the definition of $P_i$ and the choice of $\overline{c}$,  $P_i'$ is not identically zero. We will think of $P_i'$ as a polynomial in $X_{i+1}$ with the coefficients being polynomials in the variables in the set $\{Y_j : j \in S_{i+1}\}$. Now, we know that the the polynomial $P_{i+1}'$ obtained by substituting $X_{i+1}$ by  $\nwa(S_{i+1})$ is identically zero. Hence, it must be the case that $X_{i+1} - \nwa(S_{i+1})$ is a factor of $P_i'$. 

To proceed further, we need the following claim. 
\begin{claim}~\label{clm: p1}
$P_i'$ as defined above can be represented as 
$$P_i' = \sum_{r = 1}^T \prod_{j = 1}^d Q_{rj}'$$
such that each of the polynomials $Q_{rj}'$ depends on at most $s\log N$ variables. 
\end{claim}
\begin{proof}
Recall that $P$ can be represented as
$$P = \sum_{i = 1}^T \prod_{j = 1}^d Q_{ij}$$
where each $Q_{ij}$ is a polynomial in at most $s = N^{\mu}$ variables. 
In going from $P$ to $P_i'$, we have substituted each of the variables outside the set $\{Y_j : j \in S_{i+1}\} \cup \{X_{i+1}\}$ by either a constant or by the polynomial $\nwa(S_{j})|\overline{c}$ (which is a polynomial in at most $|S_j \cap S_{i+1}| \leq \log N$ variables) for some $j$. In either case, after substitution, the polynomials $Q_{rj}'$ obtained from $Q_{rj}$ depends on at most $s\log N$ variables, since $Q_{rj}$ depended on at most $s$ variables. This completes the proof of the claim.
\end{proof}
Moreover, since the individual degree of variables in $P$ is at most $k$, the individual degree of $X_{i+1}$ in $P_i'$ is at most $k$. The goal now is to invoke Lemma~\ref{lem:DSY main}, which would imply that $\nwa(S_{i+1})$ also has a small circuit as a sum of product of polynomials in {\it few} variables, and together with the lower bound from Theorem~\ref{thm:mainthm}, this would lead to a contradiction.
 We essentially follow this outline. Formally, we use the following claim to complete the proof of Lemma~\ref{lem: hitting set correctness}.  We defer the proof of the claim to the end.     
\begin{claim}~\label{clm:dsy app}
If $(X_{i+1} - \nwa(S_{i+1}) )$ divides $P_i'$, then $\nwa(S_{i+1})$ can be written as 
$$ \nwa(S_{i+1}) = \sum_{r = 1}^{I'} \prod_{j = 1}^{d'} \Gamma_{rj} $$
 where 
 \begin{enumerate}
\item $I' \leq (da'^2 + 1)\cdot {{k+a' + 1} \choose k + 1} \times {{T\cdot (k+1)^3 + a'}\choose a'}^{k+1}$
\item $d' \leq d\cdot a'$
\item Each  $\Gamma_{rj}$ depends on at most $s\log N$ variables
\end{enumerate}

\end{claim}
From our choice of parameters, recall that 
 $$a = N^{\mu/{\mu'}}\cdot \log^{1/{\mu'}} N$$
 and $$s = N^{\mu} $$

Therefore,  $s\log N \leq  N^{\mu}\cdot \log N \leq a^{\mu'}$. To complete the proof, we observe that by Theorem~\ref{thm:mainthm}, we must have $$I'd' \geq (a')^{\Omega(\sqrt{a'})}$$
But, for our choice of parameters,
\begin{enumerate}
\item $I' \leq (da'^2+1)\cdot {{k+a'} \choose k} \times {{T\cdot (k+1)^3 + a'}\choose a'}^{k+1} \leq da^{O(Tk^4)} \leq d{a'}^{O(Tk^4)}$ (since $a$ and $a'$ are polynomially related)
\item $d' \leq da'$
\end{enumerate}

This implies that $I'd' \leq d^2a^{O(Tk^4)}$. From our choice of parameters,  $s\log N < a^{\mu'}$ and $Tk^4 + 2\log d \in o(\sqrt{a'})$. This contradicts that $I'd' \geq (a')^{\Omega(\sqrt{a'})}$. This completes the proof of Lemma~\ref{lem: hitting set correctness} assuming Claim~\ref{clm:dsy app}.  
\end{proof}

We now give a proof of Claim~\ref{clm:dsy app}.
\begin{proof}[Proof of Claim~\ref{clm:dsy app}]
From Claim~\ref{clm: p1}, we know that $$P_i' = \sum_{r = 1}^T \prod_{j = 1}^d Q_{rj}'$$
such that each  $Q_{rj}'$ depends on at most $s\log N$ variables. 
Since $P_i'$ is not identically zero and $\nwa(S_{i+1})$ is a root of $P_i'$, it follows from Lemma~\ref{lem:non zero derivative} that there is an integer $\lambda$ such that  $0 \leq \lambda \leq k-1$ and,  $$\frac{\partial^{\lambda} P_i'}{\partial X_{i+1}^{\lambda}}(\nwa(S_{i+1})) = 0$$ and $$\frac{\partial^{\lambda+1} P_i'}{\partial X_{i+1}^{\lambda+1}}(\nwa(S_{i+1})) \neq 0$$

From Lemma~\ref{lem: model props} it follows that $\tilde{P_i'} = \frac{\partial^{\lambda} P_i'}{\partial X_{i+1}^{\lambda}}$ can also be expressed as $$\tilde{P_i'} = \sum_{r = 1}^{T'} \prod_{j = 1}^d \tilde{Q}_{ij}$$
where $T' \leq T\cdot (k+1)^2$ and each of the $\tilde{Q}_{rj}$ depends on at most $s\log N$ variables. 

Observe that, $\tilde{P_i'}$ vanishes when $\nwa(S_{i+1})$ is substituted for $X_{i+1}$, while its derivative with respect to $X_{i+1}$ does not vanish identically at $X_{i+1} = \nwa(S_{i+1})$. So, in particular, there is a substitution of the $Y$ variables where the derivative $\frac{\partial{\tilde{P_i'}}}{\partial{X_{i+1}}}$ is nonzero. Since the class of $\sop$ circuits is closed under translations of variables  (from item 2 in Lemma~\ref{lem: model props}), we can assume without loss of generality that the derivative is nonzero when all the variables in $\overline{Y}$ are set to zero. Also observe that by this variable translation, we have actually obtained a polynomial $\nwb(S_{i+1})$ from $\nwa(S_{i+1})$. Moreover, the degree of $\nwb(S_{i+1})$ is equal to $a'$ and the homogeneous component of degree $a'$ of $\nwb(S_{i+1})$ is equal to $\nwa(S_{i+1})$. Let the polynomial obtained after the variable translation from $\tilde{P_i'}$ as $\tilde{P_i''}$.  At this point, the hypothesis of Lemma~\ref{lem:DSY main} is satisfied by $\tilde{P_i''}$. 

Let  $\tilde{P_i''} = \sum_{j = 0}^k C_j(\overline{Y})\cdot X_{i+1}^j$. Here, $C_j(\overline{Y})$ is a polynomial only in the $Y$ variables and is the coefficient of $X_{i+1}^j$, when viewing $\tilde{P_i''}$ as an element of $\F[\overline{Y}][X_{i+1}]$. From Lemma~\ref{lem:extracting coefficients}, we know that  each of the polynomials $C_j$ can be expressed as a polynomial of the form 
$$C_j = \sum_{r= 1}^{T_j} \prod_{l = 1}^d Q_{rl}''$$
where $T_j \leq T'\cdot(k+1) \leq T\cdot (k+1)^3$ and each  $Q_{rl}''$ depends on at most $s\log N$ variables. 

Hence,  by Lemma~\ref{lem:DSY main}, for every $t \geq 0$, there exists a polynomial $R_t \in \F[Z_1, Z_2, \ldots, Z_{k+1} ]$ of degree at most $t$ such that $$\h^{\leq t}[\nwb(S_{i+1})] = \h^{\leq t}[R_t(C_0, C_1, \ldots, C_k)] $$

The goal now is to obtain a representation of $\nwa(S_{i+1})$ as a sum of products of polynomials in few variables and show that this contradicts the lower bound in Theorem~\ref{thm:mainthm}.
$\nwb(S_{i+1})$ is a  polynomial of degree at most $a'$. So,  there is a polynomial $R_{a'}$ of degree at most $a'$ in  $k + 1$ variables such that $$\nwb(S_{i+1}) =  \h^{\leq {a'}}[R_{a'}(C_0, C_1, \ldots, C_k)]$$
From the discussion on the relation between $\nwb(S_{i+1})$ from $\nwa(S_{i+1})$, we also know that 
$$\nwa(S_{i+1}) = \h^{a'}[\nwb(S_{i+1})] = \h^{a'}[R_{a'}(C_0, C_1, \ldots, C_k)]$$
Since $R_{a'}$ is a polynomial in $k+1$ variables of degree $a'$, the number of monomials in $R_{a'}$ is at most ${a' + k + 1} \choose {k+1}$. Therefore, we can represent $R_{a'}(C_0, C_1, \ldots, C_k)$ as a sum of products of the $C_j$'s, with the sum fan-in at most ${a' + k + 1} \choose {k+1}$ and the product fan-in at most $a'$. Moreover, each of the product gates in this representation takes the polynomials $C_j$'s as inputs. We know that each $C_j$ can be written as 
$$C_j = \sum_{r= 1}^{T_j} \prod_{l = 1}^d Q_{rl}''$$
where each $Q_{rl}''$ is a polynomial in at most $s\log N$ variables, and the top sum fan-in $T_j$ is at most $T\cdot (k+1)^3$. For any $t$, the polynomial $C_j^t$,  has a similar representation with the top sum fan-in at most ${T\cdot (k+1)^3 + t}\choose t$. Therefore, any product of fan-in at most $a'$ in the $C_j$'s can be written as a sum of product of polynomials in at most $s\log N$ variables, with top fan-in  at most $${{T\cdot (k+1)^3 + a'}\choose a'}^{k+1}$$
since each $C_j$ is raised to a power of at most $a'$ and there are $k+1$ such $C_j$'s.  
Therefore,  $R_{a'}(C_0, C_1, \ldots, C_k)$ can be written as $$R_{a'}(C_0, C_1, \ldots, C_k) = \sum_{r = 1}^I \prod_{j = 1}^{d'} \Gamma'_{rj}$$ such that 
\begin{enumerate}
\item $I \leq {{k+a' + 1} \choose k+1} \times {{T\cdot (k+1)^3 + a'}\choose a'}^{k+1}$
\item $d' \leq d\cdot a'$
\item Each  $\Gamma'_{rj}$ depends on at most $s\log N$ variables
\end{enumerate}
We would now like to extract the homogeneous part of degree $a'$ of $R_{a'}(C_0, C_1, \ldots, C_k)$, which we know is equal to $\nwa(S_{i+1})$. We do this by a standard application of Lemma~\ref{lem:interpolation}.  Since we are interested only in the homogeneous part of degree $a'$, we can assume without loss of generality that each of the polynomials $\Gamma'_{rj}$ is of degree at most $a'$ (we can discard all monomials of degree larger than $a'$ in each of the $\Gamma'_{rj}$, since they do not contribute to the homogeneous component of degree $a'$ of $R_{a'}(C_0, C_1, \ldots, C_k)$ ). Hence, the degree of $R_{a'}(C_0, C_1, \ldots, C_k)$ is upper bounded by $da'\cdot a'$. So, from Lemma~\ref{lem:interpolation}, we can extract the homogeneous component of degree $a'$ of $R_{a'}(C_0, C_1, \ldots, C_k)$ by blowing up the top fan-in by a factor of at most $da'^2 + 1$. Hence, $\nwa(S_{i+1})$  can be  expressed as 
 $$ \nwa(S_{i+1}) = \sum_{r = 1}^{I'} \prod_{j = 1}^{d'} \Gamma_{rj} $$
 where 
 \begin{enumerate}
\item $I' \leq (da'^2 + 1)\cdot {{k+a' + 1} \choose k + 1} \times {{T\cdot (k+1)^3 + a'}\choose a'}^{k+1}$
\item $d' \leq d\cdot a'$
\item Each  $\Gamma_{rj}$ depends on at most $s\log N$ variables
\end{enumerate}
\end{proof}

We remark that if the value of $s$ was $\log^{O(1)} N$ to start with, the same proof as above goes through with $l$ and $a$ being set to polynomials of sufficiently high degree in  $\log N$. The size of the hitting set and the time to construct it in this case are upper bounded by a quasipolynomial function in $N$. 
\section{Open problems}~\label{sec:open ques}
We  conclude with some open problems. 
\begin{enumerate}
\item An intriguing open question is to obtain PIT for $\sop$ circuits without the restriction on the individual degree. The strategy in this paper relies on  hardness randomness tradeoffs for bounded depth circuits~\cite{DSY09}.  The tradeoffs in~\cite{DSY09} crucially use the fact that the individual degree is bounded. 
\item Another related question would be to get any non-trivial PIT (even subexponential) for the sum of constant many products of degree two polynomials.
\item It would also be interesting to understand if one could obtain any non-trivial PIT for slightly non-multilinear depth four circuits (say individual degree at most 2) with bounded top fan-in. A natural strategy for this question would be to reduce it to the case of $\sop$ circuits by either expanding out the polynomials $Q_{ij}$ which depend on too many variables or use a partial derivative like trick, as in~\cite{OSV14}. The immediate challenge in this case is that the top fan-in seems to increase by any of these tricks and the calculations in this paper seem to not work out.  
\end{enumerate}

\section*{Acknowledgements}
We would like to thank Rafael Oliveira for many helpful discussions regarding hardness-randomness tradeoffs for bounded depth arithmetic circuits at the early stages of this work. 
\bibliographystyle{alpha}
\bibliography{refs}

\begin{thebibliography}{VSBR83}

\bibitem[Alo99]{AlonCN}
Noga Alon.
\newblock Combinatorial nullstellensatz.
\newblock {\em Combinatorics, Probability and Computing}, 8, 1999.

\bibitem[ASS13]{ASS13}
Manindra Agrawal, Chandan Saha, and Nitin Saxena.
\newblock Quasi-polynomial hitting-set for set-depth-\&\#916; formulas.
\newblock In {\em Proceedings of the Forty-fifth Annual ACM Symposium on Theory
  of Computing}, STOC '13, pages 321--330, New York, NY, USA, 2013. ACM.

\bibitem[ASSS12]{ASSS12}
Manindra Agrawal, Chandan Saha, Ramprasad Saptharishi, and Nitin Saxena.
\newblock Jacobian hits circuits: hitting-sets, lower bounds for depth-d
  occur-k formulas {\&} depth-3 transcendence degree-k circuits.
\newblock In {\em Proceedings of the 44th ACM symposium on Theory of
  computing}, pages 599--614, 2012.

\bibitem[AV08]{AV08}
M.~Agrawal and V.~Vinay.
\newblock Arithmetic circuits: A chasm at depth four.
\newblock In {\em FOCS}, 2008.

\bibitem[dOSV14]{OSV14}
Rafael~Mendes de~Oliveira, Amir Shpilka, and Ben~Lee Volk.
\newblock Subexponential size hitting sets for bounded depth multilinear
  formulas.
\newblock {\em Electronic Colloquium on Computational Complexity {(ECCC)}},
  21:157, 2014.

\bibitem[DSY09]{DSY09}
Zeev Dvir, Amir Shpilka, and Amir Yehudayoff.
\newblock Hardness-randomness tradeoffs for bounded depth arithmetic circuits.
\newblock {\em {SIAM} J. Comput.}, 39(4):1279--1293, 2009.

\bibitem[FLMS]{FLMS13}
H.~Fournier, N.~Limaye, G.~Malod, and S.~Srinivasan.
\newblock Lower bounds for depth 4 formulas computing iterated matrix
  multiplication.
\newblock {\em STOC 2014}.

\bibitem[For]{Forbes-personal}
Michael Forbes.
\newblock Deterministic divisibility testing via shifted partial derivatives.
\newblock {\em Personal communication}.

\bibitem[FS13a]{ForbesS13}
Michael~A. Forbes and Amir Shpilka.
\newblock Quasipolynomial-time identity testing of non-commutative and
  read-once oblivious algebraic branching programs.
\newblock In {\em 54th Annual {IEEE} Symposium on Foundations of Computer
  Science, {FOCS}}, pages 243--252, 2013.

\bibitem[FS13b]{FS13}
Michael~A. Forbes and Amir Shpilka.
\newblock Quasipolynomial-time identity testing of non-commutative and
  read-once oblivious algebraic branching programs.
\newblock {\em 2013 IEEE 54th Annual Symposium on Foundations of Computer
  Science}, 0:243--252, 2013.

\bibitem[GKKSa]{GKKS12}
A.~Gupta, P.~Kamath, N.~Kayal, and R.~Saptharishi.
\newblock Approaching the chasm at depth four.
\newblock {\em CCC 2013}.

\bibitem[GKKSb]{GKKS13}
Ankit Gupta, Pritish Kamath, Neeraj Kayal, and Ramprasad Saptharishi.
\newblock Arithmetic circuits: A chasm at depth three.
\newblock {\em In Proceedings of FOCS 2013}.

\bibitem[Gup14]{Gupta14}
Ankit Gupta.
\newblock Algebraic geometric techniques for depth-4 {PIT} {\&}
  sylvester-gallai conjectures for varieties.
\newblock {\em Electronic Colloquium on Computational Complexity {(ECCC)}},
  21:130, 2014.

\bibitem[Kay12]{Kayal12}
Neeraj Kayal.
\newblock An exponential lower bound for the sum of powers of bounded degree
  polynomials.
\newblock {\em ECCC}, 19:81, 2012.

\bibitem[KI04]{KI04}
V.~Kabanets and R.~Impagliazzo.
\newblock Derandomizing polynomial identity tests means proving circuit lower
  bounds.
\newblock {\em Computational Complexity}, 13(1-2):1--46, 2004.

\bibitem[KLSS]{KLSS14}
N.~Kayal, N.~Limaye, C.~Saha, and S.~Srinivasan.
\newblock An exponential lower bound for homogeneous depth four arithmetic
  formulas.
\newblock {\em FOCS 2014}.

\bibitem[Koi12]{koiran}
P.~Koiran.
\newblock Arithmetic circuits: The chasm at depth four gets wider.
\newblock {\em Theoretical Computer Science}, 448:56--65, 2012.

\bibitem[KS]{KS-formula}
Mrinal Kumar and Shubhangi Saraf.
\newblock The limits of depth reduction for arithmetic formulas: It's all about
  the top fan-in.
\newblock {\em STOC 2014}.

\bibitem[KS14a]{KayalSaha14}
Neeraj Kayal and Chandan Saha.
\newblock Lower bounds for depth three arithmetic circuits with small bottom
  fanin.
\newblock {\em Electronic Colloquium on Computational Complexity {(ECCC)}},
  21:89, 2014.

\bibitem[KS14b]{KS-full}
Mrinal Kumar and Shubhangi Saraf.
\newblock On the power of homogeneous depth 4 arithmetic circuits.
\newblock {\em FOCS}, 2014.

\bibitem[KSS]{KSS13}
Neeraj Kayal, Chandan Saha, and Ramprasad Saptharishi.
\newblock A super-polynomial lower bound for regular arithmetic formulas.
\newblock {\em STOC 2014}.

\bibitem[Muk]{Mukhopadhyay15}
Partha Mukhopadhyay.
\newblock Depth-4 identity testing and noether's normalization lemma.
\newblock {\em Electronic Colloquium on Computational Complexity {(ECCC)}}.

\bibitem[Nis91]{Nisan91}
Noam Nisan.
\newblock Lower bounds for non-commutative computation (extended abstract).
\newblock In {\em Proceedings of the 23rd Annual {ACM} Symposium on Theory of
  Computing, May 5-8, 1991, New Orleans, Louisiana, {USA}}, pages 410--418,
  1991.

\bibitem[NW94]{NW94}
Noam Nisan and Avi Wigderson.
\newblock Hardness vs randomness.
\newblock {\em J. Comput. Syst. Sci.}, 49(2):149--167, 1994.

\bibitem[Raz06]{Raz06}
Ran Raz.
\newblock Separation of multilinear circuit and formula size.
\newblock {\em Theory of Computing}, 2(1):121--135, 2006.

\bibitem[Sax07]{S07}
Nitin Saxena.
\newblock Diagonal circuit identity testing and lower bounds.
\newblock {\em Electronic Colloquium on Computational Complexity (ECCC)},
  14(124), 2007.

\bibitem[Sch80]{Schwartz80}
J.~T. Schwartz.
\newblock Fast probabilistic algorithms for verification of polynomial
  identities.
\newblock {\em Journal of ACM}, 27(4):701--717, 1980.

\bibitem[Shp01]{Shp01}
Amir Shpilka.
\newblock Affine projections of symmetric polynomials.
\newblock In {\em Proceedings of the 16th Annual Conference on Computational
  Complexity}, CCC '01, pages 160--, Washington, DC, USA, 2001. IEEE Computer
  Society.

\bibitem[SV11]{SarafV11}
S.~Saraf and I.~Volkovich.
\newblock Black-box identity testing of depth-4 multilinear circuits.
\newblock In {\em Proceedings of the 43rd Annual STOC}, pages 421--430, 2011.

\bibitem[SW01]{SW01}
A.~Shpilka and A.~Wigderson.
\newblock Depth-3 arithmetic circuits over fields of characteristic zero.
\newblock {\em Computational Complexity}, 10(1):1--27, 2001.

\bibitem[Tav13]{Tavenas13}
S{\'e}bastien Tavenas.
\newblock Improved bounds for reduction to depth 4 and depth 3.
\newblock In {\em MFCS}, pages 813--824, 2013.

\bibitem[Val79]{Valiant79}
L.~G. Valiant.
\newblock Completeness classes in algebra.
\newblock In {\em STOC}, 1979.

\bibitem[VSBR83]{VSBR83}
Leslie~G. Valiant, Sven Skyum, S.~Berkowitz, and Charles Rackoff.
\newblock Fast parallel computation of polynomials using few processors.
\newblock {\em SIAM Journal of Computation}, 12(4):641--644, 1983.

\bibitem[Zip79]{Zippel79}
R.~Zippel.
\newblock Probabilistic algorithms for sparse polynomials.
\newblock In {\em Symbolic and algebraic computation}, pages 216--226. 1979.

\end{thebibliography}

\appendix
\section{Calculations}~\label{sec:calc}
$$Td^2n^3 \cdot rs \geq \frac{{\frac{1}{n^{O(1)}}\text{min}\left(\frac{p^r}{4^r} \cdot {N \choose r} \cdot {N \choose m}, {N \choose m + n - r}\right)}}{{2^{O(\sqrt{n})}\cdot {N \choose m+rs} \cdot {n + r \choose r}}} $$
We first estimate the ratio $$\frac{{N \choose m + n - r}}{{{N \choose m+rs} \cdot {n + r \choose r}}} $$. 
\begin{eqnarray*}
\frac{{N \choose m + n - r}}{{{N \choose m+rs} \cdot {n + r \choose r}}} &\geq & \frac{(m+rs)!}{(m+n-r)!} \frac{(N-m-rs)!}{(N-m-(n-r))!} \cdot \left(\frac{r}{e(n+r)}\right)^r  \\
\end{eqnarray*}
Here we use the fact that ${n + r \choose r} \leq \left(\frac{e(n+r)}{r}\right)^r$.
Now, approximating the ratios using Lemma~\ref{lem:approx} and substituting $m = \frac{N}{2}(1-r\frac{\ln n}{n})$, we get 
\begin{eqnarray*}
\frac{{N \choose m + n - r}}{{{N \choose m+rs} \cdot {n + r \choose r}}} &\geq & \left( \frac{N-m}{m}\right)^{n-r-rs} \cdot \left(\frac{r}{e(n+r)}\right)^r  \\
&\geq & \exp\left(\frac{r\ln n}{n} \cdot (n-r-rs) - r\ln {\frac{e(n+r)}{r}}\right)\\
\end{eqnarray*}

Since $r = \Theta(\sqrt{n})$, we get that the ratio is at least $\exp\left(r\ln n ((n-r-rs)/n - \frac{1}{2} + o(1))\right)$, which is $\exp(\Omega(\sqrt{n}\ln n))$. \\

Next we estimate the ratio $$ \frac{{\left(\frac{p^r}{4^r} \cdot {N \choose r} \cdot {N \choose m}\right)}}{{ {N \choose m+rs} \cdot {n + r \choose r}}}  $$
\begin{eqnarray*}
\frac{{\left(\frac{p^r}{4^r} \cdot {N \choose r} \cdot {N \choose m}\right)}}{{ {N \choose m+rs} \cdot {n + r \choose r}}} &\geq & \frac{p^r}{4^r} \cdot \frac{(m+ rs)!}{m!} \cdot \frac{(N-m- rs)!}{(N-m)!} \cdot \frac{N!}{(N-r)!} \cdot \frac{n!}{(n+r)!}\\
&\geq & \frac{p^r}{4^r} \cdot \left( \frac{m}{N-m} \right)^{rs} \cdot \left( \frac{N}{n} \right)^{r}\\
 &\geq & \frac{p^r}{4^r} \cdot \left( 1 - 2.01r\frac{\ln n}{n}\right)^{rs} \cdot \left( \frac{N}{n} \right)^{r}\\
 &\geq &\frac{1}{4^r} \exp\left(-r(\mu + \delta)\ln N - 2.01r^2s\frac{\ln n}{n} + r \ln (N/n) \right)\\
\end{eqnarray*}
Here, we used Lemma~\ref{lem:approx} in the second step and substituted $p = N^{-(\delta + \mu)}$ in the last step. 
Now, substituting $2n^{2+\gamma} \geq N \geq n^{2+\gamma} $, the exponent is at least $$r\ln n (-(\mu + \delta)(2 + \gamma) - 2.01rs/n + (1 + \gamma))$$ 
This is at least $$r\ln n (-(\mu + \delta)(2 + \gamma) - 2.01rs/n+ (1 + \gamma))$$
Now, plugging back the value of $\gamma$, the exponent is at least $(2 - 2.01rs/n)r\ln n$. We have chosen $rs$ such that $rs/n < 0.001$. Therefore, the ratio we set out to lower bound is at least $\exp(\Omega(\sqrt{n}\ln n))$. 
\end{document}